\documentclass[runningheads]{llncs}
\usepackage{multirow}    %
    \usepackage{xcolor}
    \usepackage{graphicx}
\usepackage{amsmath}
\usepackage{subfig}
\usepackage{nicefrac}       
\usepackage{algorithm}
\usepackage{algorithmicx}
\usepackage[noend]{algpseudocode}
\usepackage{subfig,booktabs,makecell}
\usepackage[export]{adjustbox}
\usepackage{hyperref}
\hypersetup{
    colorlinks=true,
    linkcolor=blue,
    filecolor=magenta,
    urlcolor=cyan,
}

\newcommand{\mdp}{\textit{MDP}}
\newcommand{\Aspace}{\ensuremath{\mathcal{A}}}
\newcommand{\Sspace}[1]{\ensuremath{\mathcal{S_{#1}}}}
\newcommand{\pdf}[1]{f_{#1}}

\newcommand{\Per}[1]{\textbf{$\mathcal{P}_{#1}$}}
\newcommand{\Exp}[1]{\textbf{$\mathcal{E}_{#1}$}}
\newcommand{\Uni}[2]{\textbf{$\mathcal{U}_{#1,#2}$}}
\newcommand{\Norm}[2]{\textbf{$\mathcal{N}_{#1,#2}$}}
\newcommand{\Q}{\ensuremath{\mathsf{QFlip}}} 
\newcommand{\phase}[1]{R_{#1}}
\newcommand{\A}{\textit{A}}
\newcommand{\lm}[1]{\ensuremath{\mathsf{LM_{#1}}}}
\newcommand{\LM}[1]{LM}
\newcommand{\cost}[1]{k_{#1}}
\newcommand{\localgain}[1]{\ensuremath{\Gamma_{#1}^{(\action{t},\state{t})}}}
\newcommand{\localben}[1]{\ensuremath{\beta_{#1}^{(\action{t},\state{t})}}}
\newcommand{\actualval}[1]{\ensuremath{V_{\state{#1},\action{#1}}}}
\newcommand{\state}[1]{\ensuremath{s_{#1}}}
\newcommand{\action}[1]{\ensuremath{a_{#1}}}
\newcommand{\rew}[1]{\ensuremath{r_{#1}}}

\newcommand{\oppLM}{\ensuremath{{\mathsf{oppLM}}}}
\newcommand{\ownLM}{\ensuremath{{\mathsf{ownLM}}}}
\newcommand{\composite}{\ensuremath{{\mathsf{composite}}}}

\newcommand{\myparagraph}[1]{\smallskip \noindent \textbf{#1.}}

\DeclareMathOperator*{\argmax}{argmax}

\newcommand\flipit{\ensuremath{\mathsf{FlipIt}}}
\newcommand{\ignore}[1]{}

\begin{document}
\title{\Q{}: An Adaptive Reinforcement Learning Strategy for the \flipit{} Security Game\thanks{will appear in \textit{10th Conference on Decision and Game Theory in Security}}
}
\titlerunning{\Q{}: An Adaptive Reinforcement Learning Strategy}
%

\author{Lisa Oakley and
Alina Oprea}

\authorrunning{}
%
\institute{Khoury College of Computer Sciences, Northeastern University, Boston MA, USA}
%
\maketitle
\begin{abstract}

A rise in Advanced Persistent Threats (APTs) has introduced a need for robustness against long-running, stealthy attacks which circumvent existing cryptographic security guarantees.
\flipit\ is a security game that models attacker-defender interactions in advanced scenarios such as APTs. Previous work analyzed extensively non-adaptive strategies in \flipit, but adaptive strategies rise naturally in practical interactions as players receive feedback during the game. We model the \flipit\ game as a Markov Decision Process and introduce \Q{}, an adaptive strategy for \flipit{} based on temporal difference reinforcement learning. We prove theoretical results on the convergence of our new strategy against an opponent playing with a Periodic strategy. We confirm our analysis experimentally by extensive evaluation of \Q{} against specific opponents. \Q{} converges to the optimal adaptive strategy for Periodic and Exponential opponents using associated state spaces.  Finally, we introduce a generalized \Q{} strategy with composite state space that outperforms a Greedy strategy for several distributions including Periodic and Uniform, without prior knowledge of the opponent's strategy. We also release an OpenAI Gym environment for \flipit{} to facilitate future research.



\end{abstract} 
\section{Introduction}

Motivated by sophisticated cyber-attacks such as Advanced Persistent Threats (APT), the \flipit\ game was introduced by van Dijk et al. as a model of cyber-interactions in APT-like scenarios~\cite{FlipIt}. \flipit\ is a two-player cybersecurity game in which the attacker and defender contend for control of a sensitive resource (for instance a password, cryptographic key, computer system, or network). Compared to other game-theoretical models, \flipit\ has the unique characteristic of \emph{stealthiness}, meaning that players are not notified about the exact state of the resource during the game. Thus, players need to schedule  moves during the game with minimal information about the opponent's strategy. The challenge of determining the optimal strategy is in finding the best move times to take back resource control, while at the same time minimizing the overall number of moves (as players pay a cost upon moving). \flipit\ is a repeated, continuous game, in which players can move at any time and benefits are calculated according to the asymptotic control of the resource minus the move cost.


The original \flipit\ paper performed a detailed analysis of non-adaptive strategies in which players move according to a renewal process selected at the beginning of the game. Non-adaptive strategies are randomized, but do not benefit from feedback received during the game. In the real world, players  naturally get information about the game and the opponent's strategy as play progresses. For instance, if detailed logging and monitoring is performed in an organization, an attacker might determine the time of the last key rotation or machine refresh upon system takeover. van Dijk et al. defined \emph{adaptive strategies} that consider various amounts of information received during gameplay, such as the time since the last opponent move. However, analysis and experimentation in the adaptive case has remained largely unexplored. In a theoretical inspection, van Dijk et al. prove that the optimal Last Move adaptive strategy against Periodic and Exponential opponents is a Periodic strategy. They also introduce an adaptive Greedy strategy that selects moves to maximize local benefit. However, the Greedy strategy requires extensive prior knowledge about the opponent (the exact probability distribution of the renewal process), and does not always result in the optimal strategy~\cite{FlipIt}. Other extensions of \flipit\ analyzed modified versions of the game~\cite{FlipThem,FlipLeakage,TestMove,ThreePlayer1,ThreePlayer2}, but mostly considered non-adaptive strategies.

In this paper, we tackle the challenge of analyzing the two-player \flipit\ game with one adaptive player and one non-adaptive renewal player.  We limit the adaptive player's knowledge  to the opponent's last move time and show how this version of the game can be modeled as an agent interacting with a Markov Decision Process (MDP). We then propose for the first time the use of temporal difference reinforcement learning for designing adaptive strategies in the \flipit\ game. We introduce \Q{}, a Q-Learning based adaptive strategy that plays the game by leveraging information about the opponent's last move times. We explore in this context the instantiation of various reward and state options to maximize \Q{}'s benefit against a variety of opponents.

We start our analysis by considering an opponent playing with the Periodic with random phase strategy, also studied by \cite{FlipIt}. We demonstrate for this case that \Q{} with states based on the time since opponent's last move converges to the optimal adaptive strategy (playing immediately after the opponent with the same period). We provide a theoretical analysis of the convergence of \Q{} against this Periodic opponent. Additionally, we perform detailed experiments in the OpenAI Gym framework, demonstrating fast convergence for a range of parameters determining the exploration strategy and learning decay. Next, we perform an analysis of \Q{} against an Exponential opponent, for which van Dijk et al. determined the optimal strategy~\cite{FlipIt}. We show experimentally that \Q{} with states based on the player's own move converges to the optimal strategy and the time to convergence depends largely on the adaptive player's move cost and the Exponential player's distribution parameters. Finally, we propose a generalized, composite \Q{} instantiation that uses as state the time since last moves for both players. We show that composite \Q{} converges to the optimal strategy for Periodic and Exponential. Remarkably, \Q{} has no prior information about the opponent strategy at the beginning of the game, and most of the time outperforms the Greedy algorithm (which leverages information about the opponent strategy). For instance, \Q{} achieves average benefit between 5\% and 50\% better than Greedy against Periodic and 15\% better than Greedy against a Uniform player.


The implications of our findings are that reinforcement learning is a promising avenue for designing optimal learning-based strategies in cybersecurity games. Practically, our results also reveal that protecting systems against adaptive adversaries is a difficult task and defenders need to become adaptive and agile in face of advanced attackers. To summarize, our contributions in the paper are:

\begin{itemize}
\item We model the \flipit\ game with an adaptive player competing against a renewal opponent as an MDP.
\item We propose \Q{}, a versatile generalized Q-Learning based adaptive strategy for \flipit\ that does not require prior information about the opponent strategy.
\item We prove \Q{} converges to the optimal strategy against a Periodic opponent.
\item We demonstrate experimentally that \Q{} converges to the optimal strategy and outperforms the Greedy strategy for a range of opponent strategies.
\item We release an OpenAI Gym environment for \flipit{} to aid future researchers.
\end{itemize}

\paragraph{Paper organization.} We start with surveying the related work in Section~\ref{sec:related}. Then we introduce the \flipit\ game in Section~\ref{sec:background} and describe our MDP modeling of \flipit\ and the \Q{} strategy in Section~\ref{sec:strategy}. We analyze \Q{} against a Periodic opponent theoretically in Section~\ref{sec:qperiodic}. We perform experimental evaluation of Periodic and Exponential strategies in Section~\ref{sec:qexp}. We evaluate generalized composite \Q{} against four distributions in Section~\ref{sec:general}, and conclude in Section~\ref{sec:conclusions}.

\section{Related Work}
\label{sec:related}
\flipit{}, introduced by van Dijk et al.~\cite{FlipIt}, is a non-zero-sum cybersecurity game where two players compete for control over a shared resource. The game distinguishes itself by its stealthy nature, as moves are not immediately revealed to players during the game. Finding an optimal (or dominant) strategy in \flipit\ implies that a player can schedule its defensive (or attack) actions most effectively against stealthy opponents. van Dijk et al. proposed multiple non-adaptive  strategies and proved results about their strongly dominant opponents and Nash Equilibria~\cite{FlipIt}. They also introduce the Greedy adaptive strategy and show that it results in a dominant strategy against Periodic and Exponential players, but it is not always optimal. The original paper left many open questions about designing general adaptive strategies for \flipit. van Dijk et al.~\cite{FlipItApp} analyzed the applications of the game in real-world scenarios such as password and key management.


Additionally, several \flipit{} extensions have been proposed and analyzed.
These extensions focus on modifying the game itself, adding additional players~\cite{ThreePlayer1,ThreePlayer2}, resources~\cite{FlipThem}, and move types~\cite{TestMove}. FlipLeakage considers a version of \flipit\  in which information leakage is gradual and ownership of resource is obtained incrementally~\cite{FlipLeakage}. Zhang et al. consider limited resources with an upper bound on the frequency of moves and analyze Nash Equilibria in this setting~\cite{ResourceConstraint}.  Several games study human defenders players against automated attackers using Periodic strategies~\cite{Human1,Human2,Human3}. All of this work uses exclusively non-adaptive players, often limiting analysis to solely opponents playing periodically. The only previous work that considers adaptive strategies is by Laszka et al.~\cite{Noncovert1,Noncovert2}, but in a modification of the original game with non-stealthy defenders. \Q{} can generalize to play adaptively in these extensions, which we leave to future work.

Reinforcement learning (RL) is an area of machine learning in which an agent takes action on an environment, receiving feedback in the form of a numerical reward, and adapting its action policy over time in order to maximize its cumulative reward. Traditional methods are based primarily on Monte Carlo and temporal difference Q-Learning \cite{Sutton}. Recently, approximate methods based on deep neural networks have proved effective at complex games such as Backgammon, Atari and AlphaGo~\cite{TDGammon,Atari,AlphaGo}.

RL has emerged in security games in recent years. Han et al. use RL for adaptive cyber-defense in a Software-Defined Networking setting and consider adversarial poisoning attacks against the RL training process~\cite{SDN}. Hu et al. proposes the idea of using Q-Learning as a defensive strategy in a cybersecurity game for detecting APT attacks in IoT systems~\cite{IoT}. Motivated by HeartBleed, Zhu et al. consider an attacker-defender model in which both parties synchronously adjust their actions, with limited information on their opponent~\cite{Heartbleed}.  Other RL applications include network security~\cite{NetworkSecurity1},  spatial security games~\cite{Spatial}, security monitoring~\cite{Monitor}, and crowdsensing~\cite{Crowd}. Markov modeling for moving target defense has also been proposed~\cite{MTD1,MTD2}.

To the best of our knowledge, our work presents the first application of RL to stealthy security games, resulting in the most effective adaptive \flipit{} strategy. 
\section{Background on the \flipit{} Game}
\label{sec:background}

\flipit\ is a two-player game introduced by van Dijk et al. to model APT-like scenarios~\cite{FlipIt}. In \flipit, players move at any time to take control of a resource. In practice, the resource might correspond to a password, cryptographic key, or computer system that both attacker and defender wish to control. Upon moving, players pay a move cost (different for each player). Success in the game is measured by player benefit, defined as the asymptotic amount of resource control (gain) minus the total move cost as described in Figure \ref{fig:notation}. The game is infinite, and we consider a discrete version of the game in which players can move at discrete time ticks. Figure \ref{fig:flipit} shows an example of the \flipit\ game, and Figure \ref{fig:notation} provides relevant notation we will use in the paper.

\begin{figure}
    \centering
    \includegraphics[width=.85\textwidth]{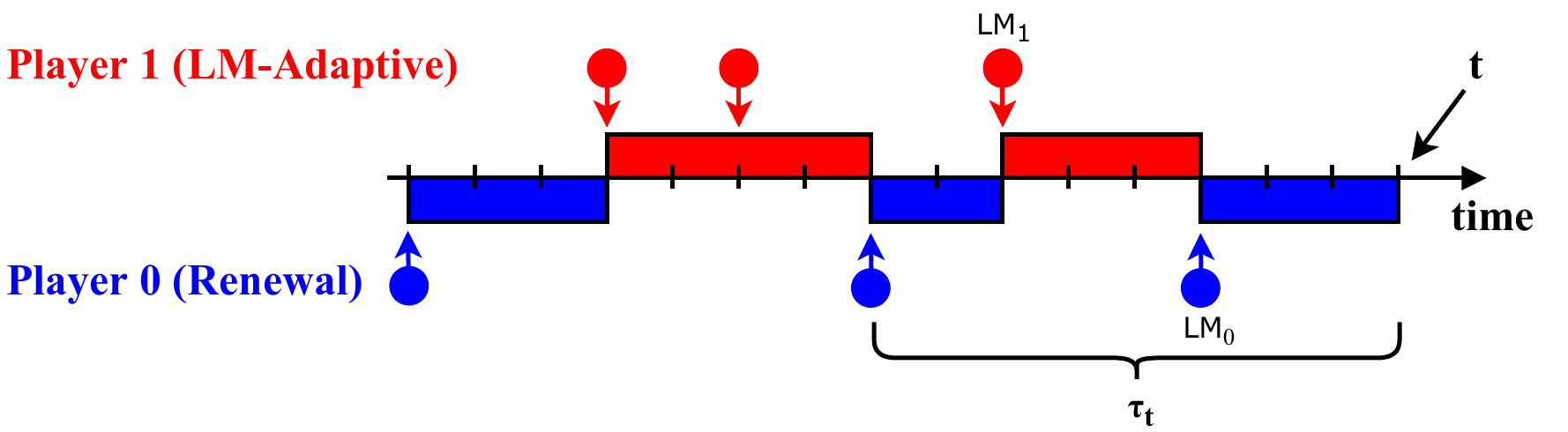}
    \caption{Example of \flipit{} game between Last Move adaptive Player 1 and Player 0 using a Renewal strategy. Rounded arrows indicate player moves. The first move of Player 1 is \emph{flipping}, and the second move is \emph{consecutive}. $\tau_t$ is the time since Player 0's last \textit{known} move at time $t$ and \lm{i} is Player $i$'s \textit{actual} last move at time $t$. Due to the stealthy nature of the game, $\tau_t\geq t-\lm{0}$.}
    \label{fig:flipit}
\end{figure}

\begin{figure}[t]
    \captionsetup[subfigure]{labelformat=empty}
        \centering
        \subfloat[]{\adjustbox{width=0.6\columnwidth,valign=B,raise=2\baselineskip}{%
        \begin{tabular}{|c|l|}
            \hline
            Symbol     & Description                                                                                                       \\   \hline
            $t$ & Time step (tick) \\
            $\cost{i}$ & Player $i$'s move cost                                                                                        \\
            $\Gamma_i$ & Player $i$'s total gain    (time in control)                                                                \\
            $n_i$ & Player $i$'s total moves \\
            $\beta_i$ & Player $i$'s total benefit.     $\beta_i=\Gamma_i-\cost{i} \cdot n_i$ \\
            $\tau_{t}$ & Time since opponent's last     known move at time $t$ \\
            $\lm{i}$ & Player $i$'s actual last     move time at time $t$ \\ 
            $\rho$ & Player 0's average move time \\\hline
            \end{tabular}}}
    \qquad
    \subfloat[]{\adjustbox{width=0.325\columnwidth,valign=B,raise=2\baselineskip}{%
        \begin{tabular}{|c|l|}
            \hline
            Symbol     & Description                                                                                                       \\   \hline
            
            \state{t} & Observation \\
            \action{t} & Action \\
            \rew{t} & Reward \\
            $\gamma$ & Future discount                               \\
            $\alpha$ & Count of \action{t} in \state{t}                                                               \\
            $\epsilon$ & Exploration parameter                                                               \\
            $d$ & Exploration discount                                                               \\
            $p$ & New move probability \\
             \hline
            \end{tabular}}}
    \setlength{\abovecaptionskip}{-15pt}
    \caption{\flipit{} notation (left) and \Q{} notation (right)}\label{fig:notation}
\end{figure}

An  interesting aspect of \flipit\ is that the players do not automatically learn the opponent's moves in the game. In other words, moves are stealthy, and players need to move without knowing the state of the resource. There are two main classes of strategies defined for \flipit:

\paragraph{Non-adaptive Strategies.} Here, players do not receive any  feedback upon moving. Non-adaptive strategies are determined at the beginning of the game, but they might employ randomization to select the exact move times. \emph{Renewal strategies} are non-adaptive strategies that generate the intervals between consecutive moves according to a renewal process. The inter-arrival times between moves are independent and identically distributed random variables chosen from a probability density function (PDF). Examples of renewal strategies include:

\begin{itemize}
\item Periodic with random phase (\Per{\delta}):  The player first moves uniformly at random with phase $\phase{\delta}\in(0,\delta)$, with each subsequent move occurring periodically, i.e., exactly at $\delta$ time units after the previous move.

\item Exponential: The inter-arrival time is distributed according to an exponential (memoryless) distribution $\Exp{\lambda}$ with rate $\lambda$. The probability density function for \Exp{\lambda} is $\pdf{\Exp{\lambda}}(x)= \lambda e^{-\lambda x}$, for $x>0$, and 0 otherwise.

\item Uniform: The inter-arrival time is distributed according to an uniform distribution \Uni{\delta}{u} with parameters $\delta$ and $u$. The probability density function for \Uni{\delta}{u} is
$\pdf{\Uni{\delta}{u}}(x)=1/u$, for  $x \in [\delta-u/2,\delta+u/2]$, and 0 otherwise.

\item Normal: The inter-arrival time is distributed according to a normal distribution \Norm{\mu}{\sigma} with mean $\mu$ and standard deviation $\sigma$. The probability density function for \Norm{\mu}{\sigma} is $\pdf{ \Norm{\mu}{\sigma}}(x)=\frac{1}{\sqrt{2 \pi \sigma^2}}e^{-(x-\mu)^2/2\sigma^2}$, for $x \in R$.
\end{itemize}

\paragraph{Adaptive Strategies.} In these strategies, players receive feedback during the game and can adaptively change their subsequent moves. In \emph{Last Move (LM)} strategies, players receive information about the opponent's last move upon moving in the game. This is the most restrictive and therefore most challenging subset of adaptive players, so we only focus on LM adaptive strategies here.

Theoretical analysis of the optimal LM strategy against specific Renewal strategies has been shown~\cite{FlipIt}. For the Periodic strategy, the optimal LM strategy is to move right after the Periodic player (whose moves can be determined from the LM feedback received during the game). The memoryless property of the exponential distribution implies that the probability of moving at any time is independent of the time elapsed since the last player's move. Thus, an LM player that knows the Exponential opponent's last move time has no advantage over a non-adaptive player. Accordingly, the dominant LM strategy against an Exponential player is still a Periodic strategy, with the period depending on the player's move cost and the rate of the Exponential opponent.



\paragraph{Greedy Strategy.} To the best of our knowledge, the only existing adaptive strategy against general Renewal players is the ``Greedy'' strategy~\cite{FlipIt}. Greedy calculates the ``local benefit'', $L(z)$ of a given move time, $z$, as:

\begin{equation}
L(z) = \frac{1}{z} \Bigl[\int_{x=0}^z x \hat{f_0}(x) dx + z \int_{z}^{\infty}\hat{f}_0(x)dx - k_1 \Bigr],
\label{eqn:localben}
\end{equation}

\noindent where $\hat{f_0}(x) = f_0(\tau+x)/(1-F_0(\tau))$, $f_0$ is the probability density function (PDF) of the opponent's strategy, $F_0$ is the corresponding cumulative density function (CDF), and $\tau$ is the interval since the opponent's last move. Greedy finds the move time, $\hat{z}$, which maximizes this local benefit, and schedules a move at $\hat{z}$ if the maximum local benefit is positive. In contrast, if the local benefit is negative, Greedy chooses not to move, dropping out of the game.


Although the Greedy strategy is able to compete with any Renewal strategy, it is dependent on prior knowledge of the opponent's strategy. van Dijk et al. showed that Greedy can play optimally against Periodic and Exponential players~\cite{FlipIt}. However, they showed a strategy for which Greedy is not optimal. This motivates us to look into other general adaptive strategies that \emph{achieve higher benefit than Greedy} and \emph{require less knowledge about the opponent's strategy}.

\section{New Adaptive Strategy for \flipit}\label{sec:strategy}


Our main insight is to apply traditional reinforcement learning (RL) strategies to the \flipit\ security game to create a Last Move adaptive strategy that outperforms existing adaptive strategies. We find that modeling \flipit\ as a \textit{Markov Decision Process} (\mdp{}) and defining an \LM{} Q-Learning strategy is non-trivial, as the stealthy nature of the game resists learning. We consider the most challenging setting, in which the adaptive player has no prior knowledge on the opponent's strategy. In this section, we present \Q{}, a strategy which is able to overcome those challenges and elegantly compete against any Renewal opponent.

\subsection{Modeling \flipit{} as an MDP}\label{sec:MDP}

Correctly modeling the game of two-player \flipit\ as an \mdp{} is as important to our strategy's success as the RL algorithm itself. In our model, player 1 is an \textit{agent} interacting with an \textit{environment} defined by the control history of the \flipit\ resource as depicted in Figure \ref{fig:mdp}.
\begin{figure}[tb]
    \centering
    \includegraphics[width=.9\textwidth]{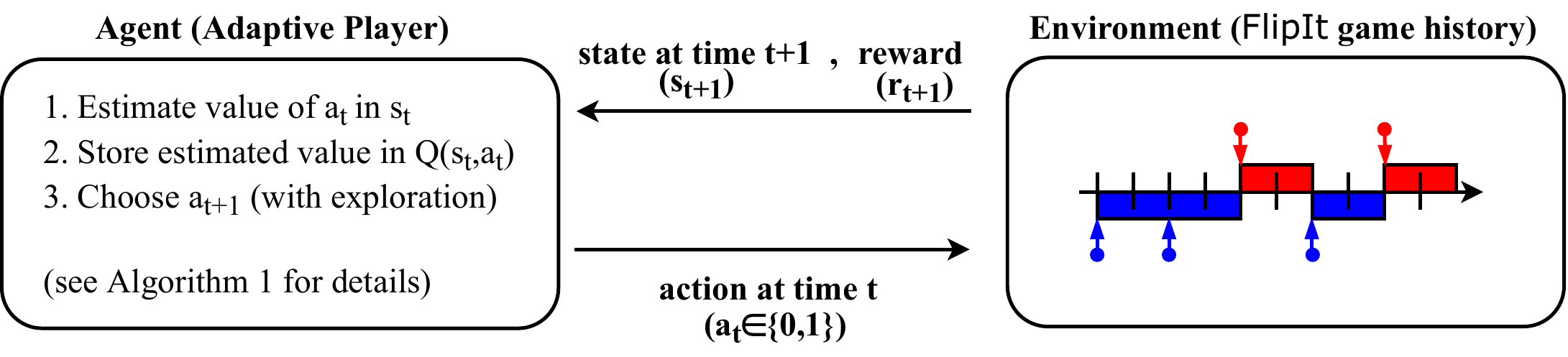}
    \caption{Modeling \flipit\ as an MDP.}
    \label{fig:mdp}
  \end{figure}

We consider the infinite but discrete version of \flipit\, and say that at every time step (tick), $t\in\{1,2,\dots\}$, the game is in some state $\state{t}\in\Sspace{}$ where \Sspace{} is a set of observed state values dependent on the history of the environment. At each time $t$, the agent chooses an action $\action{t} \in\Aspace{}=\{0,1\}$ where $0$ indicates waiting, and $1$ indicates moving. The environment updates accordingly and sends the agent state $\state{t+1}$ and reward $\rew{t+1}$ defined in Table \ref{fig:state-types} and Equation (\ref{eqn:rew}), respectively.

Defining optimal state values and reward functions is essential to generating an effective RL algorithm. In a stealthy game with an unknown opponent, this is a non-trivial task that we will investigate in the following paragraphs.

\myparagraph{Modeling State}\label{sec:state}
At each time step $t$, the \LM{} player knows two main pieces of information: its own last move time ($\lm{1}$), and the time since the opponent's last \textit{known} move ($\tau_t$). The observed state can therefore depend on one or both of these values. We define three observation schemes in Table \ref{fig:state-types}. We compare these observation schemes against various opponents in the following sections.

\begin{table}[tb]
    \centering
    \caption{Observation Schemes}\label{tab1}
    \begin{tabular}{|l|l|l|}
    \hline
    Scheme&\state{0}&\state{t} for $t>0$\\
    \hline
    \oppLM{}&$-1$&$\tau_t$ if $\lm{1}>$ player 0 first move, otherwise $-1$\\
    \ownLM{}&$0$&$t-\lm{1}$\\
    \composite{}&$(\state{0}^{\ownLM{}},\state{0}^{\oppLM{}})$&$(\state{t}^{\ownLM{}},\state{t}^{\oppLM{}})$\\
    \hline
    \end{tabular}
    \label{fig:state-types}
  \end{table} 

\myparagraph{Modeling Reward}\label{sec:rew}
Temporal difference learning algorithms leverage incremental updates from rewards at each time tick, and therefore require the environment to transmit meaningful reward values. A good reward should be flexible, independent of prior knowledge of opponent strategy, and most importantly promote the ultimate goal of achieving high benefit. We divide actions into three resulting categories: \textit{flipping}, \textit{consecutive}, and \textit{no-play} based on the type of action and the state of the environment as depicted in Table \ref{fig:move-types}. 


\begin{table}[tb]
  \centering
  \caption{Player 1 Action Categories}\label{tab1}
  \begin{tabular}{|l|l|l|l|l|l|}
  \hline
   Move type & \action{t} & Env State & Cost  & Outcome & Explanation\\
  \hline
  \textit{flipping}&1&$\lm{0}>\lm{1}$&$-\cost{1}$&$\tau_{t+1} = t - \lm{0} + 1$&Player 1 takes control\\
  \textit{consecutive}&1&$\lm{0}\leq \lm{1}$&$-\cost{1}$&$\tau_{t+1}=\tau_{t}+1$&Player 1 moves while in control\\
  \textit{no-play}&0&any&0&$\tau_{t+1}=\tau_{t}+1$&Player 1 takes no action\\
  \hline
  \end{tabular}\label{fig:move-types}
\end{table}

The most straightforward reward after each action is a \textit{resulting benefit}
\begin{equation}
  \localben{1}=\localgain{1}-\cost{1}
\end{equation}
where \localgain{1} is the \textit{resulting gain}, or Player 1's additional time in control between time $t$ and the opponent's next move as a result of taking action \action{t} in state \state{t}. These rewards would sum to equal Player 1's total benefit over the course of the game, therefore exactly matching the goal of maximized benefit.

For \textit{consecutive} moves,
$\localben{1}=-\cost{1}$, as Player 1 is already in control, and therefore attains no additional gain, resulting in a wasted move.
For \textit{no-plays}, $\localben{1}=0$, as not moving guarantees no additional gain and incurs no cost.

The main challenge here comes from determining reward for \textit{flipping} moves. Consider the case where the opponent plays more than once between two of the agent’s moves. Here, it is impossible to calculate an accurate \localben{1}, as the agent cannot determine the exact time they lost control. Moreover, there is no way to calculate future gain from any move against a randomized opponent, as the opponent's next move time is unknown.

We acknowledge a few ineffective responses to these challenges. The first rewards Player 1 for playing soon after the opponent (higher reward for lower resulting $\tau_{t+1}$ values). This works against a Periodic opponent, but not work against an Exponential opponent as it does not reward optimal play. Another approach is a reward based on prior gain, rather than resulting gain. This is difficult to calculate and rewards previous moves, rather than the current action.

We determined experimentally that the best reward for \action{t} against an unknown opponent is a fixed constant related to the opponent's move frequency as follows
\begin{equation}
  \rew{t+1} =
    \begin{cases}
      0 & \text{if \textit{no-play} at time }  t \\
      -\cost{1} & \text{if \textit{consecutive} move at time } t\\
      \frac{\rho-\cost{1}}{c} & \text{if \textit{flipping} move at time } t
   \end{cases}\label{eqn:rew}
\end{equation}
\noindent where $\rho$ is an estimate of Player 0's average move frequency and $c$ is a constant determined before gameplay (for normalization).
Playing often toward the beginning of the game and keeping track of the observed move times can provide a rough estimate of $\rho$. This reward proves highly effective, while maintaining the flexibility to play against any opponent without any details of their strategy.

\subsection{The \Q{} Strategy}\label{sec:Q}
In this section we present a new, highly effective \LM{} adaptive strategy, \Q{}, based on existing temporal difference reinforcement learning techniques. \Q{} plays within our \flipit{} model from Section \ref{sec:MDP}. Though optimized to play against Renewal opponents, a \Q{} player can compete against any player, including other adaptive opponents. To the best of our knowledge, \Q{} is the first adaptive strategy which can play \flipit{} against both Renewal and non-Renewal opponents without any prior knowledge about their strategy.

\myparagraph{Value Estimation}
\Q{} uses feedback attained from the environment after each move and the information gathered during gameplay to estimate the value of action \action{t} in state \state{t}. Player 1 has no prior knowledge of the opponent's strategy, therefore must learn an optimal strategy in real-time. We adopt an on-line temporal difference model of value estimation where $Q(\state{t},\action{t})$ is the expected value of taking action \action{t} in state \state{t} as in~\cite{Sutton}.

We start by defining the actual value of an action \action{t} in state \state{t} as a combination of the immediate reward and potential future value as
\begin{equation}
    \actualval{t}=\rew{t+1} + \gamma\cdot \max_{\action{}'\in \Aspace{}}Q(\state{t+1},\action{}').\label{eqn:v}
\end{equation}
where $0\leq\gamma\leq 1$ is a constant discount to the estimated future value, and \rew{t+1} is the environment-provided reward from Equation (\ref{eqn:rew}).

After each tick, we update our value estimate by a discounted difference between estimated move value and actual move value as follows:
\begin{equation}
    Q_{\alpha+1}(\state{t},\action{t}) = Q_{\alpha}(\state{t},\action{t}) + \frac{1}{\alpha+1}(\actualval{t}- Q_{\alpha}(\state{t},\action{t}))\label{eqn:estimate}
\end{equation}where $\alpha$ is the number of times action $a$ has been performed in state $s$, and $1/\alpha$ is the \textit{step-size} parameter, which discounts the change in estimate proportionally to the number of times this estimate has been modified.

This update policy uses the estimate error to step toward the optimal value estimate at each state.  Note that, if $\gamma=0$, we play with no consideration of the future, and $Q_\alpha(s,a)$ is just an average of the environment-provided rewards over all times action \action{} was performed in state \state{}. 

\myparagraph{Action Choice (Exploration)}
A key element of any reinforcement learning algorithm is balancing exploitation of learned  value estimation with exploration of new states and actions to avoid getting stuck in local maxima. We employ a modified \textit{decaying-}$\epsilon$\textit{-greedy} exploration strategy from~\cite{Sutton} as

\begin{equation}
  \action{t} =
  \begin{cases}
    \text{choose uniformly at random from }\Aspace{}\text{ with probability }\epsilon'\\
    \argmax_{a\in\Aspace{}} Q(\state{t},a)\text{ with probability }1-\epsilon'\label{eqn:epsgreedy}
  \end{cases}
\end{equation}
where $\epsilon'=\epsilon\cdot e^{-d\cdot v}$ for constant exploration and decay parameters $0\leq\epsilon,d<1$ and $v$ equal to the number of times \Q{} has visited state $\state{t}$. If $Q(\state{t},0)=Q(\state{t},1)$, we choose $\action{t}=0$ with probability $p$, and $\action{t}=1$ with probability $1-p$.

\myparagraph{Algorithm Definition} We then define the agent's policy in Algorithm \ref{alg:q-policy}. This is a temporal difference Q-Learning based algorithm. The algorithm first estimates the opponent move rate, $\rho$, by playing several times. This step is important to determine if it should continue playing or drop out (when $\cost{1} \ge \rho$), and to fix the environment reward. If the agent decides to play, it proceeds to initialize the Q table of estimated rewards in each state and action to 0.   The agent's initial state is set according to Table \ref{fig:state-types}. The action choice is based on exploration, as previously discussed. Once an action is selected, the agent receives the reward and new state from the environment and updates Q according to Equation (\ref{eqn:estimate}).


\begin{algorithm}
  \caption{}
	\label{alg:q-policy}
  \begin{algorithmic}[1]
    \State Estimate rate of play $\rho$ of opponent and drop out if $\cost{1} \ge \rho$
    \State Initialize 2D table Q with all zeros
    \State Initialize $\state{0}$ according to observation type (see section \ref{sec:state})
    \For{$t\in\{1,2,\dots\}}$
      \If{Q(\state{t},0)=Q(\state{t},1)}
        \State $\action{t} \gets 0$ with probability $p$, else $\action{t} \gets 1$
      \Else
        \State Choose action $\action{t}$ according to equation (\ref{eqn:epsgreedy})
      \EndIf
      \State Simulate action $\action{t}$ on environment, and observe \state{t+1}, \rew{t+1}
      \State Update Q(\state{t},\action{t}) according to equation (\ref{eqn:estimate})
    \EndFor
  \end{algorithmic}
\end{algorithm} 

\section{Theoretical Analysis for Periodic Opponent}
\label{sec:qperiodic}

We consider first an opponent playing periodically with random phase. Previous work has focused primarily on analyzing \Per{\delta} strategies in a non-adaptive context~\cite{FlipIt,Noncovert1,Noncovert2,ThreePlayer2}.
We first show theoretically that \Q{} eventually learns to play optimally against a Periodic opponent when the future discount $\gamma$ is set at 0.
We employ this restriction because, when $\gamma>0$, the actual value of \action{t} in \state{t} (\actualval{t}) depends on the maximum estimated value in \state{t+1}, which changes concurrently with $Q(\state{t},\action{t})$.
Additionally, Section \ref{sec:per_experiment} shows experimentally that changing $\gamma$ does not have much effect on benefit.

In the discrete version of \flipit\, a player using the Periodic strategy \Per{\delta} plays first at some uniformly random $\phase{\delta}\in\{0,\dots,\delta\}$, then plays every $\delta$ ticks for all subsequent moves. In this case, the optimal \LM{} strategy is to play immediately after the opponent. Our main result is the following theorem, showing that \Q{} converges to the optimal strategy with high probability.

\begin{theorem}
    Playing against a \Per{\delta} opponent with $\gamma=0$, $\cost{1}<\delta$, \Q{} using the \ownLM{} observation scheme converges to the optimal LM strategy as $t\rightarrow\infty$.
    \label{thm:opt-per_optimal}
\end{theorem}

We will prove this theorem by first showing that \Q{} visits state $\delta+1$ infinitely often, then claiming that \Q{} will eventually play once in state $\delta+1$. We conclude by proving \Q{} eventually learns to play in state $\delta+1$ and no other state. This is exactly the optimal strategy of playing right after the opponent. Because the \Per{\delta} strategy is deterministic after the random phase, we can model Player 1's known state and transitions according to the actual state of the game as in Figure \ref{fig:states}. We prove several lemmas and finally the main theorem below.

\begin{figure}
  \centering
  \includegraphics[width=\textwidth]{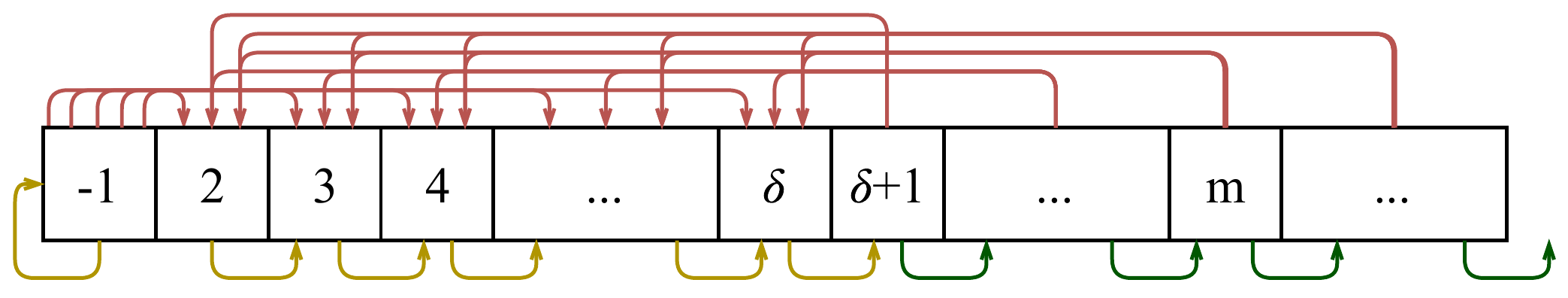}
  \caption{\Q{} using \oppLM{} states against \Per{\delta}. Arrows indicate state transitions. Red is $\action{t}=1$, green is $\action{t}=0$, and yellow is either $\action{t}=0$ or $\action{t}=1$.}\label{fig:states}
\end{figure}

\begin{lemma}
  If $t>\phase{\delta}$ and $\action{t}=1$, \Q{} will visit state $\delta+1$ in at most $\delta$ additional time steps.\label{lem:q_state_lt_delta}
\end{lemma}
\begin{proof}
  We want to show that, for any $\state{t}$, with $t>\phase{\delta}$, choosing $\action{t}=1$ means that \Q{} will visit state $\delta+1$ again in at most $\delta$ additional time steps.

  \textit{Case 1}: Assume $1<\state{t}<\delta+1$. We see from Figure \ref{fig:states} that $\state{t+1}=\state{t}+1$, for all $\action{t}$ when $-1<\state{t}<\delta+1$. Therefore, \Q{} will reach state $\delta+1$ in $\delta+1-\state{t}<\delta$ additional time steps.

  \textit{Case 2}: Assume $\state{t}\geq\delta+1$. The opponent is \Per{\delta}, so $\lm{1}<\lm{0}$ at time $t$. Given $\action{t}=1$, Table \ref{fig:move-types} gives $\state{t+1}=t-\lm{0}+1<\delta+1$, returning to Case 1.

  \textit{Case 3}: Assume $\state{t}=-1$ and \Q{} chooses $\action{t}=1$. From Tables \ref{fig:state-types} and \ref{fig:move-types}, we have that $\state{t+1}=t-\lm{0}<\delta+1$, returning again to Case 1.\qed
\end{proof}

\begin{lemma}
  Playing against a \Per{\delta} opponent with $\gamma=0$, $0\leq p< 1$ and $\cost{1}<\delta$, \Q{} visits state $\delta+1$ infinitely often.
  \label{lem:q_reach_delta}
\end{lemma}
\begin{proof}
  We prove by induction on the number of visits, $n$, to state $\delta+1$ that \Q{} visits state $\delta+1$ infinitely often.

  \textbf{Base of Induction:}
 We show that, starting from $\state{0}=-1$, \Q{} will reach $\state{t}=\delta+1$ with probability converging to 1. If \Q{} chooses $\action{t}=1$ when $\phase{\delta}<t\leq\delta$, this is a \textit{flipping} move, putting \Q{} in state $\state{t+1}<\delta+1$. If not, $t>\delta\geq\phase{\alpha}$. When $\state{t}=-1$, \Q{} chooses $\action{t}=0$ with probability $p$ and $\action{t}=1$ with probability $1-p$. Therefore
  \begin{equation}
    P[\text{\Q{} not flipping after } v\text{ visits to } \state{}=-1]=p^v.
  \end{equation}
  This implies that \Q{} will flip with probability $1-p^v \rightarrow  1$ as $t=v\rightarrow\infty$. By Lemma \ref{lem:q_state_lt_delta}, \Q{} will reach state $\delta+1$ in finite steps with probability 1.

  \textbf{Inductive Step:} Assume \Q{} visits state $\delta+1$ $n$ times. Because we are considering an infinite game, at any time $t$ there are infinitely many states that have not been visited. Therefore $\exists m\geq\delta+1$ such that $\forall s> m,Q(s)=(0,0)$.

  If \Q{} flips at state $\state{t}\in\{\delta+1,...,m\}$, it will reach $\delta+1$ again in a finite number of additional steps by Lemma \ref{lem:q_state_lt_delta}.

  If \Q{} does not flip at state $\state{t}\in\{\delta+1,...,m\}$, we have
  \begin{equation}
    P[\text{\Q{} does not move after } z\text{ steps}]=p^z\label{eqn:noplay0}
  \end{equation} since probability of moving when $Q(\state{})=(0,0)$ is $p$ and not moving implies $\state{t+1}=\state{t}+1>m$. Therefore, as $t=z\rightarrow\infty$, \Q{} will flip again with probability $1-p^z\rightarrow 1$.

  By mathematical induction, we have that state $\delta+1$ is visited infinitely often with probability converging to $1$.\qed
\end{proof}

\begin{lemma}
  If $\gamma=0$ and $\cost{1}<\delta$, \Q{} will eventually choose to move in state $\delta+1$ with probability 1.\label{lem:q_flip}
\end{lemma}
\begin{proof}
  If \Q{} flips in any visit to state $\delta+1$, the conclusion follows.

  Assume \Q{} does not flip in state $s=\delta+1$.
  Since $\gamma=0$, from Equation (\ref{eqn:v}), $\actualval{t}=0$ for $\action{t}=0$. Therefore $Q(\state{})=(0,0)$ and we have
  \begin{equation}
    P[\Q{} \text{ does not flip after }v\text{ visits to state }\delta+1]=p^v.
  \end{equation}
 By Lemma \ref{lem:q_reach_delta}, we know that \Q{} visits state $\delta+1$ infinitely often. Therefore, probability that \Q{} moves in state $\delta+1$ is $1-p^v\rightarrow 1$ as $v\rightarrow\infty$.\qed
\end{proof}

\subsubsection{Proof of Theorem \ref{thm:opt-per_optimal}} We will now prove the original theorem, using these lemmas.
    To prove that \Q{} plays optimally against \Per{\delta}, we must show that it will eventually (1) play at $\state{}=\delta+1$ at each visit and (2) not play at $\state{}  \not =\delta+1$ at any visit.
    Assuming $\gamma=0$, we have from section 2.5 of~\cite{Sutton} that
    \begin{equation}
        Q_{\alpha+1}(s,a)=Q_\alpha(s,a)+\frac{1}{\alpha+1}\cdot (\rew{\alpha+1}-Q_\alpha(s,a))=\frac{1}{\alpha+1}\sum_{\alpha+1}^{i=1}\rew{i}.\label{eqn:sutton}
    \end{equation}
    Here we denote by $r_i$ the reward obtained the $i$-th time state $s$ was visited and action $a$ was taken. Additionally, \Per{\delta} plays every $\delta$ time steps after the random phase. Therefore we derive from Equation (\ref{eqn:rew}):
    \begin{equation}
        \rew{i} =
          \begin{cases}
            0 & \text{if }\action{i}=0 \\
            -\cost{1} & \text{if } \action{i}=1 \text{ and } 1\leq\state{i}\leq\delta\\
            \frac{\delta-\cost{1}}{c} & \text{if } \action{i}=1 \text{ and } \state{i}\geq\delta+1
         \end{cases}\label{eqn:rew_pf}
      \end{equation}

    By Equations (\ref{eqn:sutton}) and (\ref{eqn:rew_pf}), we have for all states s,
    \begin{equation}
        Q_\alpha(s,0)=\frac{1}{\alpha}\sum^{\alpha}_{i=1}0=0.
    \end{equation}

    First we show that \Q{} will eventually choose $\action{t}=0$ in all states $1<\state{t}<\delta+1$. Consider some \state{t} such that $1<\state{}<\delta+1$. If \Q{} never chooses $\action{t}=1$ in this state, we are done. Assume \Q{} plays at least once in this state, $\alpha>0$, then
    \begin{equation}
        Q_\alpha(s,1)=\frac{1}{\alpha}\sum^{\alpha}_{i=1}-\cost{1}=-\cost{1}<0=Q_\alpha(s,0)
    \end{equation}
    since $\cost{1}>0$. Therefore $\argmax_{a\in\Aspace{}}{Q(s,a)}=0$. Because $\epsilon'=\epsilon\cdot e^{-d\cdot v}$, and $0\leq\epsilon,d<1$, as $v\rightarrow \infty$ we have that $\epsilon'\rightarrow 0$. Therefore $P[\Q{}$ does not play at $\state{}] \rightarrow 1$ for $1\leq\state{}\leq\delta$ as desired.

    Next we show that \Q{} will eventually play at state $\delta+1$ at each visit. From Lemma \ref{lem:q_flip}, we know that \Q{} will play once at $\state{}=\delta+1$ with probability $1$, meaning $\alpha>0$ with probability $1$. By Equations (\ref{eqn:sutton}) and (\ref{eqn:rew_pf}) we have $\text{for } \alpha>0 \text{ and }\state{}=\delta+1.$
    \begin{equation}
        Q_\alpha(s,1)=\frac{1}{\alpha}\sum^{\alpha}_{i=1}\frac{\delta-\cost{\A{}}}{c}=\frac{\delta-\cost{\A{}}}{c}>0=Q_\alpha(s,0).
    \end{equation}
     Now $\argmax_{a\in\Aspace{}}Q(s,a)=1$, so as $\epsilon'\rightarrow 0$, $P[\Q{}$ plays at $\state{}=\delta+1] \rightarrow 1$.

     If \Q{} plays at state $\delta+1$, it will not reach states $s>\delta+1$, and thus cannot play in those states. Therefore, as $t\rightarrow\infty$, $P[\text{\Q{} plays optimally}]\rightarrow1$.

    \qed

\section{\Q{} Against \Per{\delta} and \Exp{\lambda} Opponents}
\label{sec:qexp}

In this section we show experimentally that \Q{} learns an optimal strategy against \Per{\delta} and \Exp{\lambda} opponents. Prior to starting play, we allow \Q{} to choose its observation scheme based on the opponent's strategy. \Q{} chooses the \oppLM{} observation against \Per{\delta} and the \ownLM{} observation against \Exp{\lambda}, but sets other parameters of \Q{} identically. This reflects the theoretical analysis from~\cite{FlipIt} which states that an optimal adaptive strategy against a  \Per{\delta} opponent depends on $\tau$, while $\tau$ is irrelevant in optimal strategies against an \Exp{\lambda} opponent.  Next section, we generalize \Q{} to play with no knowledge of opponent strategy.

\subsection{Implementation}

All simulations (\url{https://github.com/lisaoakley/flipit-simulation}) are written in 
Python 3.5  with a custom OpenAI Gym environment for \flipit{} (\url{https://github.com/lisaoakley/gym-flipit}).  We ran each experiment over a range of costs and Player 0 parameters within the constraint that $\cost{1}<\rho$, as other values have an optimal drop out strategy. For consistency, we calculated $\rho$ from the distribution parameters before running simulations. We report averages across multiple runs, choosing number of runs and run duration to ensure convergence. Integrals are calculated using the \textit{scipy.integrate.quad} function. For Greedy's maximization step we used \textit{scipy.optimize.minimize} with the default ``BFGS'' algorithm. \Q{} can run against a variety of opponents with minimal configuration, thus we set all \Q{} hyper-parameters identically across experiments, namely $\gamma=0.8$, $\epsilon=0.5$, $d=0.05$, $c=5$, and $p=0.7$ unless otherwise noted.

\begin{figure}[t]
  \centering
  \subfloat[][$8,000$ ticks]{\includegraphics[width=.289\textwidth]{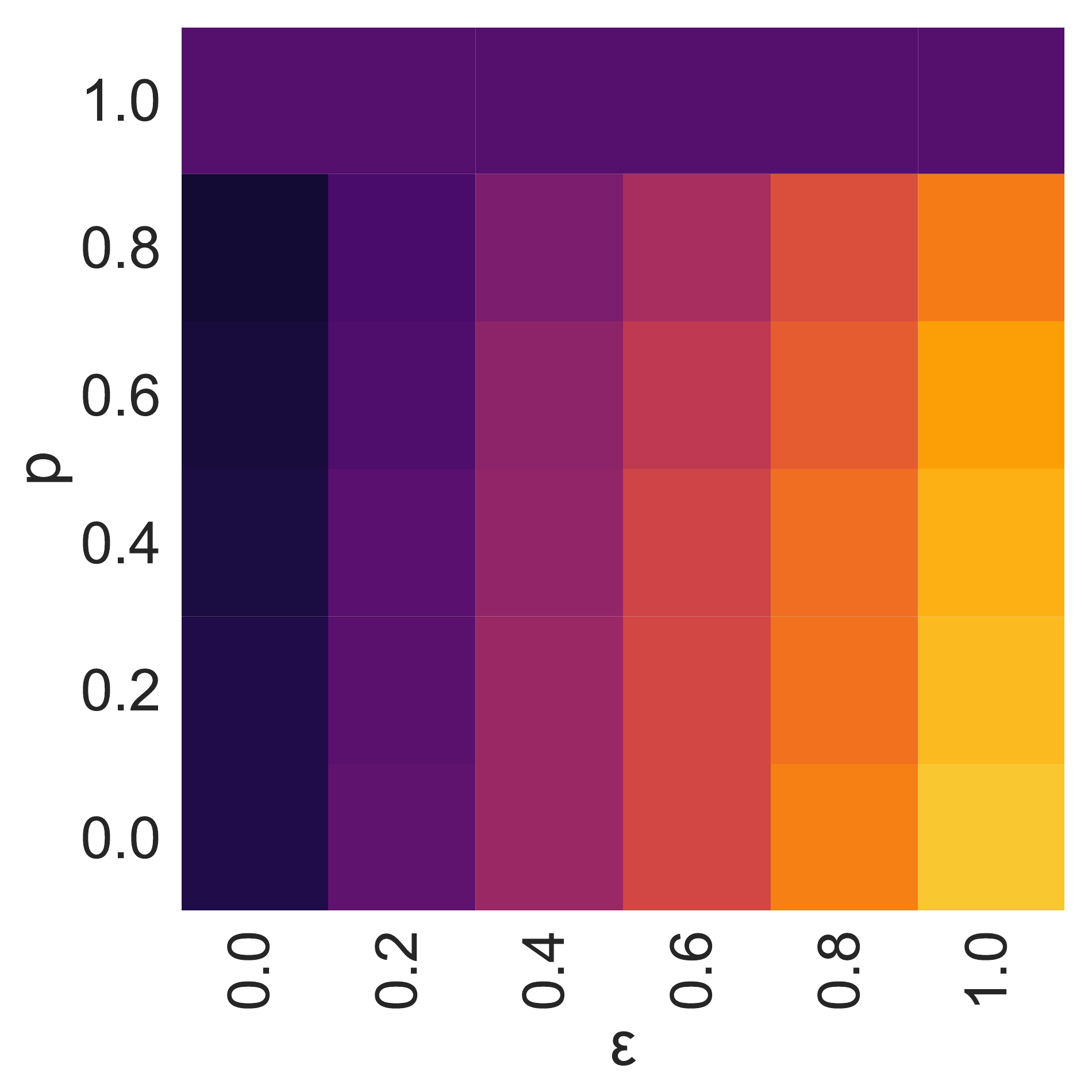}}\hspace{1em}%
  \hfill
  \subfloat[][$16,000$ ticks]{\includegraphics[width=.289\textwidth]{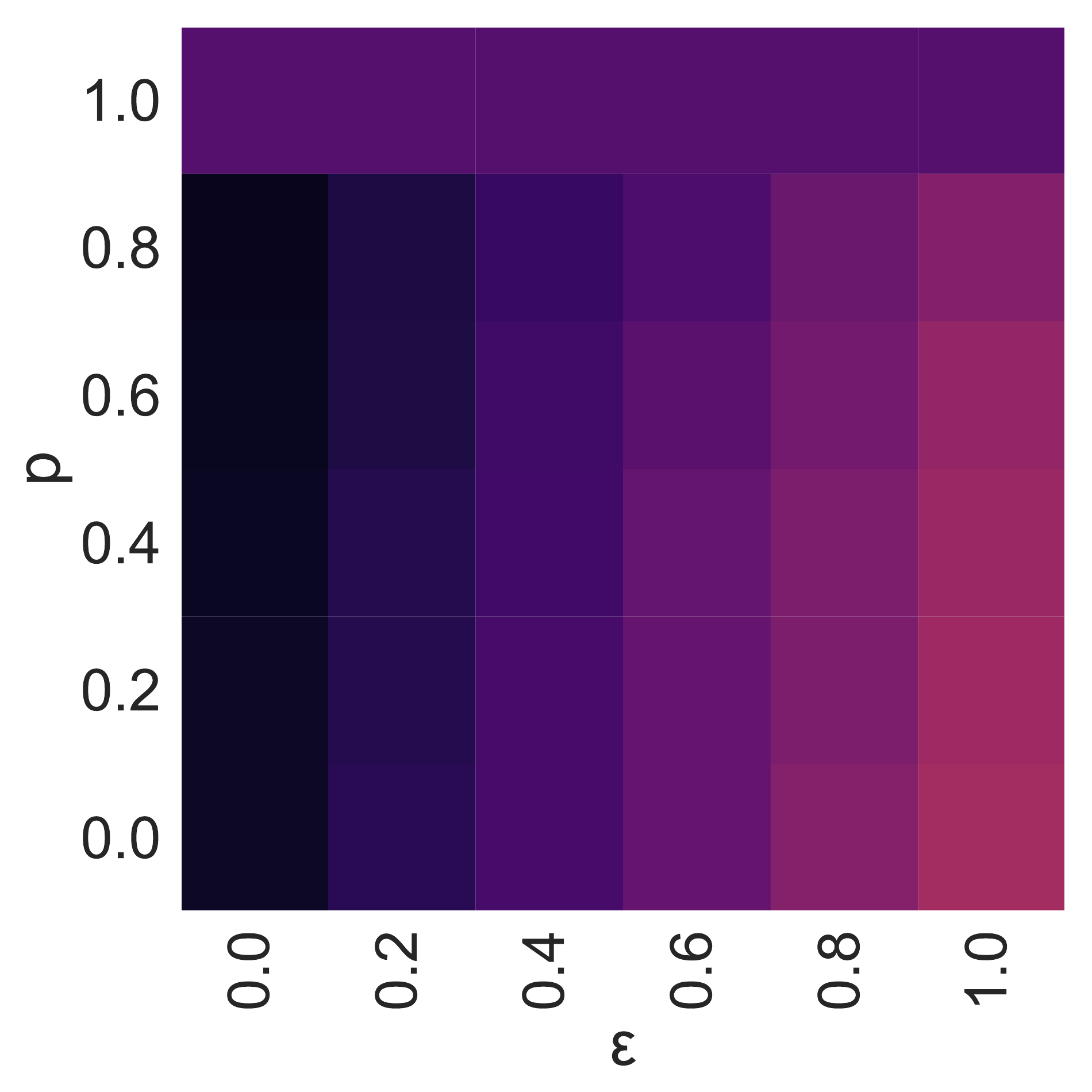}}\hspace{1em}%
  \hfill
  \subfloat[][$64,000$ ticks]{\includegraphics[width=.366\textwidth]{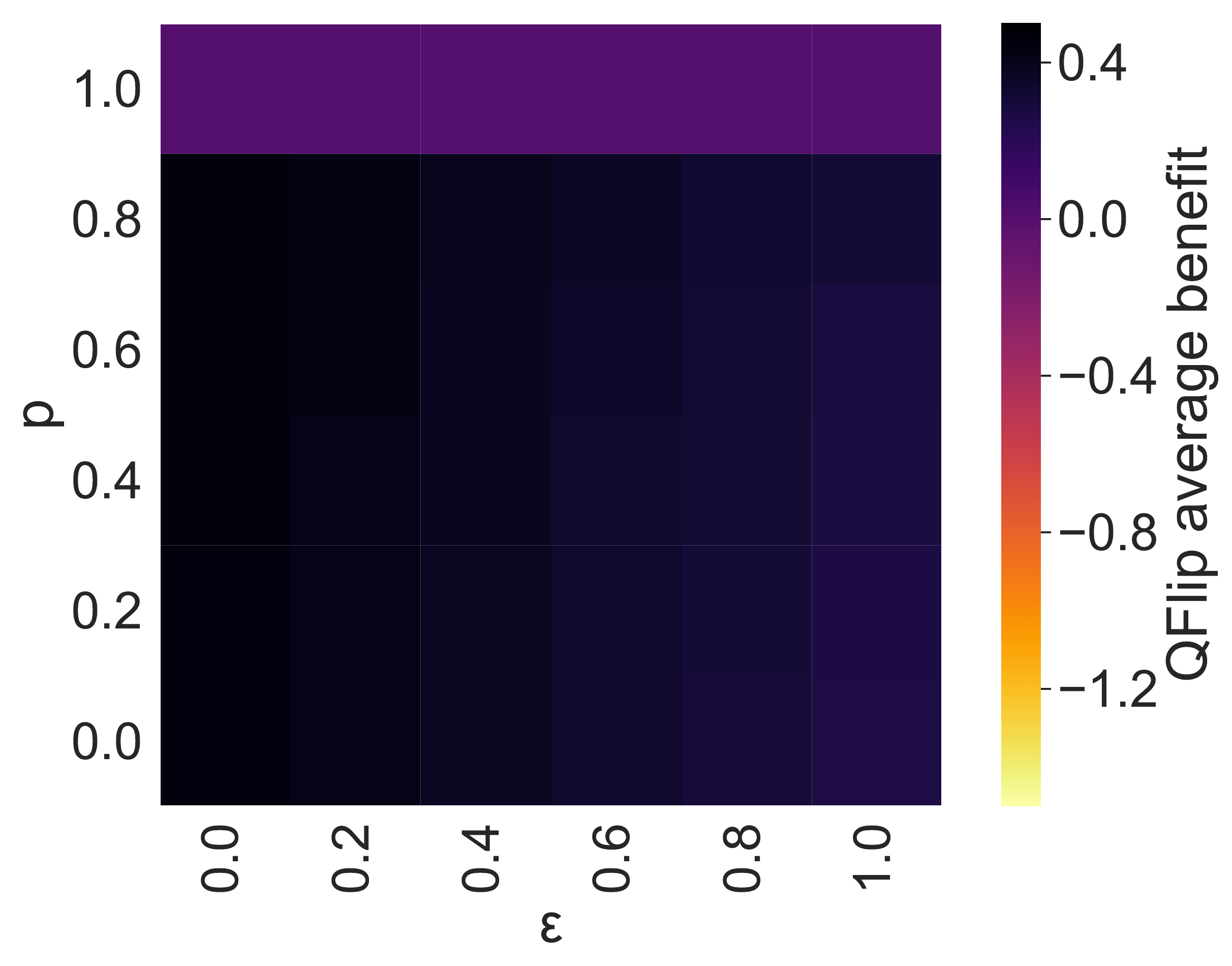}}
  \caption{Player's 1 average benefit for different $\epsilon$ and $p$ parameters at three time ticks.  Here \Q{} plays with \oppLM{}, $\gamma=0$ and $\cost{1}=25$ against a \Per{\delta} opponent with $\delta=50$, averaged over $10$ runs. Darker purples mean higher average benefit. \Q{} converges to optimal ($.48$), improving more quickly with low exploration.}
  \label{fig:tick_heat}
\end{figure}

\subsection{\Q{} vs. Periodic}\label{sec:per_experiment}
In Section \ref{sec:qperiodic} we proved that \Q{} will eventually play optimally against \Per{\delta} when future discount $\gamma$ is $0$ and $\cost{1}<\delta$. In this section, we show experimentally that most configurations of $\Q{}$ using $\oppLM{}$ quickly lead to an optimal strategy. Additionally, we show there is little difference in benefit when $\gamma>0$.

\myparagraph{Learning Speed} When $\gamma=0$, Equation (\ref{eqn:v}) reduces to $\actualval{t}=\rew{t+1}$. In this case we verify that \Q{} learns an optimal strategy for all exploration parameters, but that lower exploration rates cause \Q{} to reach optimal benefit more quickly. Against a Periodic opponent, $Q_\alpha(s,a)$ is constant after \Q{} takes action \action{} in state \state{} at least once ($\alpha>0$). Thus, exploring leads to erroneous moves. When the probability of moving for the first time in state $s$ is low ($p$ is high), \Q{} makes fewer costly incorrect moves in states $s<\delta+1$ leading to higher benefit. When $p=1$, \Q{} never plays and $\beta_1=0$.  Figure \ref{fig:tick_heat} displays Player 1's average benefit for different values of $p$ and $\epsilon$. \Q{} achieves close to optimal benefit after $64,000$ ticks with high exploration rates $\epsilon$ and high probability of playing in new states ($1-p$), and in as little as $8,000$ ticks with no exploration ($\epsilon=0$) and low $1-p$.


\myparagraph{Varied Configurations} When $\gamma>0$, \actualval{t} factors in the estimated value of state \state{t+1} allowing \Q{} to attain positive reward for choosing not to move in state $\delta+1$. The resulting values in the Q table can negatively impact learning. Figure \ref{fig:expl_gamma} shows the average benefit over time  (left) and the number of non-optimal runs (right) for different $\epsilon$ and $\gamma$ values. We observe that \Q{} performs non-optimally on $38\%$ of runs with $\gamma=0.8$ and $\epsilon=0$  but has low benefit variation between runs. However, increasing $\epsilon$ even to $0.1$ compensates for this and allows \Q{} to play comparably on average with future estimated value ($\gamma>0$) as with no future estimated value ($\gamma=0$). This result allows us flexibility in configuring \Q{} which we will leverage to maintain hyper-parameter consistency against all opponents. In the rest of the paper we set $\gamma=0.8$, $\epsilon=0.5$, and $p=0.7$.

\begin{figure}[t]
  \centering
    \subfloat[width=.6\textwidth][Learning over time averaged over 50 runs per configuration. Optimal benefit vs. \Per{\delta} with $\delta=50$ and $\cost{1}=25$ in a discrete game is $0.48$]{\includegraphics[width=.55\textwidth]{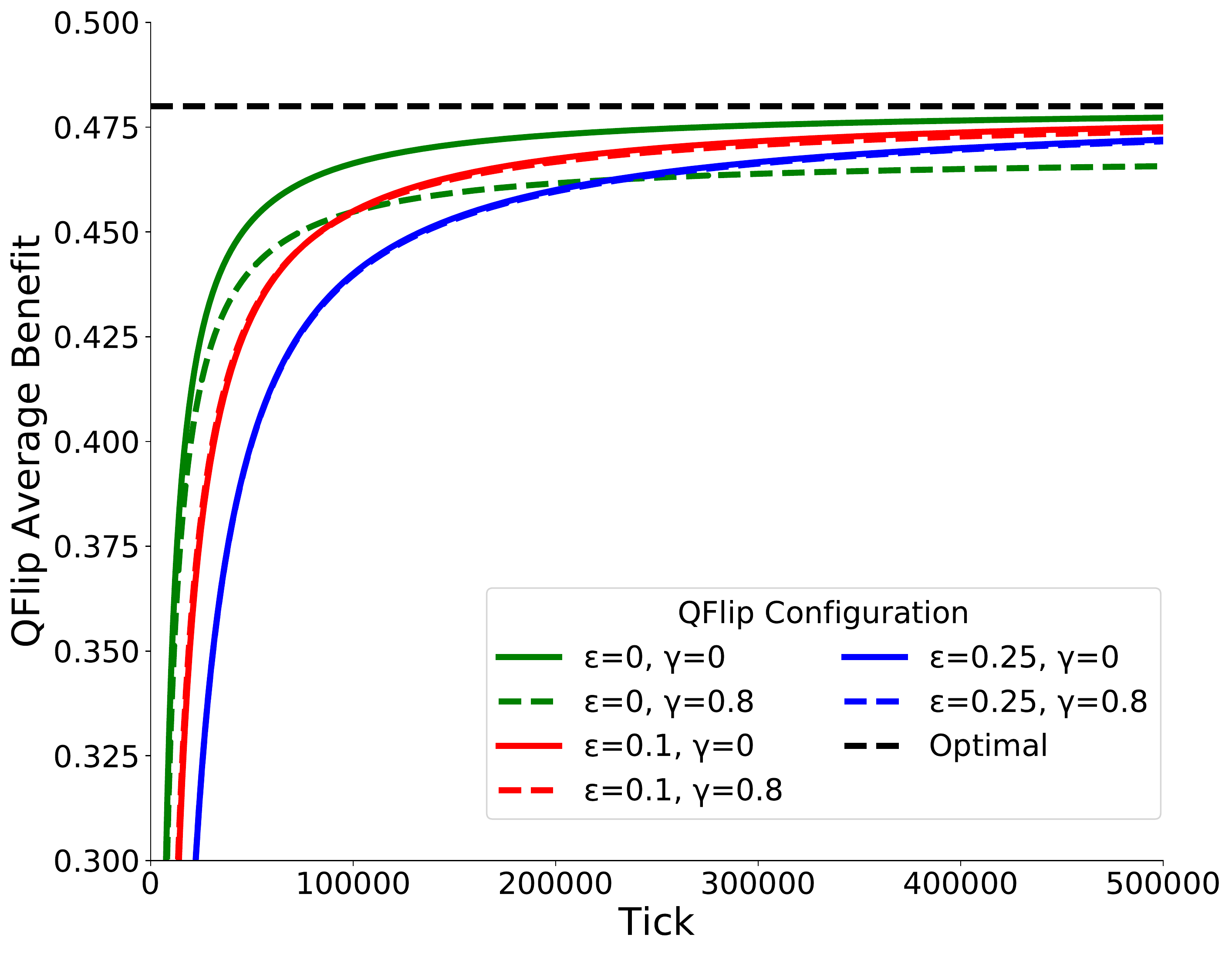}}
    \qquad
    \subfloat[Statistics over 50 runs. ``Non-optimal" runs have average benefit $>.02$ less than optimal (.48) after 500,000 ticks.]{\adjustbox{width=0.38\columnwidth,valign=B,raise=2\baselineskip}{%
    \begin{tabular}{|c|c|c|c|c|}

      \hline
      $\gamma$             & $\epsilon$ & \begin{tabular}[c]{@{}c@{}}\# \\ non-optimal \\ runs\end{tabular} & \begin{tabular}[c]{@{}c@{}}min \\ benefit\end{tabular} & \begin{tabular}[c]{@{}c@{}}max \\ benefit\end{tabular} \\ \hline
      \multirow{4}{*}{0}   & 0          & 0                                                                 & 0.477                                                      & 0.478                                                      \\ \cline{2-5}
                           & 0.1        & 0                                                                 & 0.474                                                      & 0.476                                                      \\ \cline{2-5}
                           & 0.25       & 0                                                                 & 0.471                                                      & 0.473                                                      \\ \cline{2-5}
                           & 0.5        & 0                                                                 & 0.465                                                      & 0.468                                                      \\ \hline
      \multirow{4}{*}{0.8} & 0          & 19                                                                 & 0.417                                                      & 0.478                                                      \\ \cline{2-5}
                           & 0.1        & 3                                                                 & 0.455                                                      & 0.476                                                      \\ \cline{2-5}
                           & 0.25       & 1                                                                 & 0.452                                                      & 0.473                                                      \\ \cline{2-5}
                           & 0.5        & 0                                                                 & 0.465                                                      & 0.468                                                      \\ \hline
      \end{tabular}
              }}
    \caption{\Q{} using \oppLM{} with $\cost{1}=25$ playing against \Per{\delta} with $\delta=50$.}
    \label{fig:expl_gamma}
\end{figure}


\myparagraph{Comparison to Greedy}
Assuming the Greedy strategy against $\Per{\delta}$ plays first at time $\delta$, it will play optimally with probability $1-\cost{1}/\delta$. However, with probability $\cost{1}/\delta$, Greedy will drop out after its first adaptive move~\cite{FlipIt}. We compare \Q{} and Greedy against \Per{\delta} for $\delta=50$ across many costs in Figure \ref{fig:greedy_vs_q}. \Q{} consistently achieves better average benefit across runs, playing close to optimally on average. Additionally, Player 0 with $\cost{0}=1$ attains more benefit on average against a Greedy opponent as a result of these erroneous drop-outs. For $\cost{1}<45$, \Q{} attains benefit between 5\% and 50\% better than Greedy on average.

\begin{figure}[t]
    \centering
    \subfloat[width=.45\textwidth][Player 1 average benefit by cost]{\includegraphics[width=.45\textwidth]{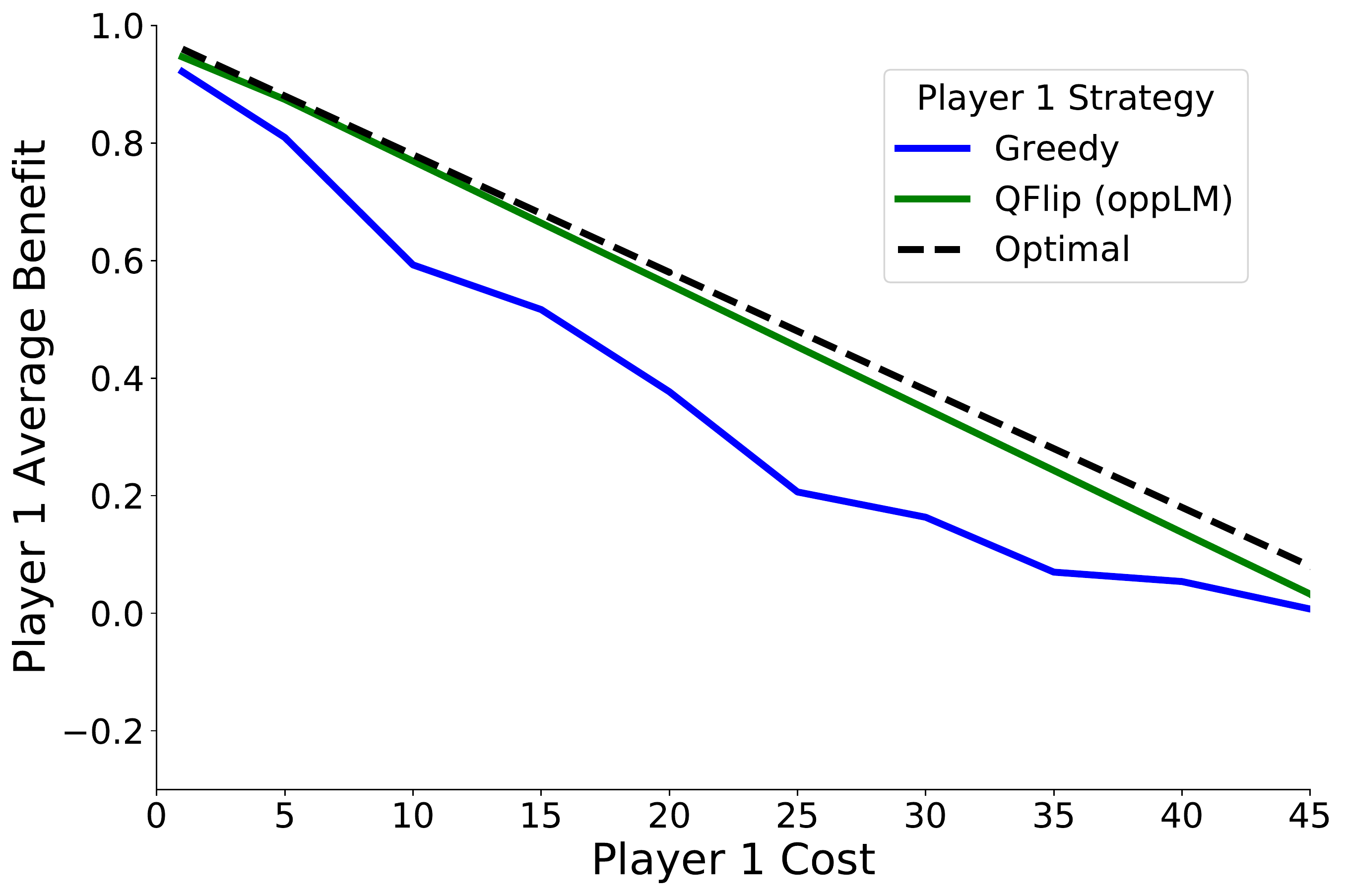}}
    \qquad
    \subfloat[width=.45\textwidth][Player 0 average benefit by cost]{\includegraphics[width=.45\textwidth]{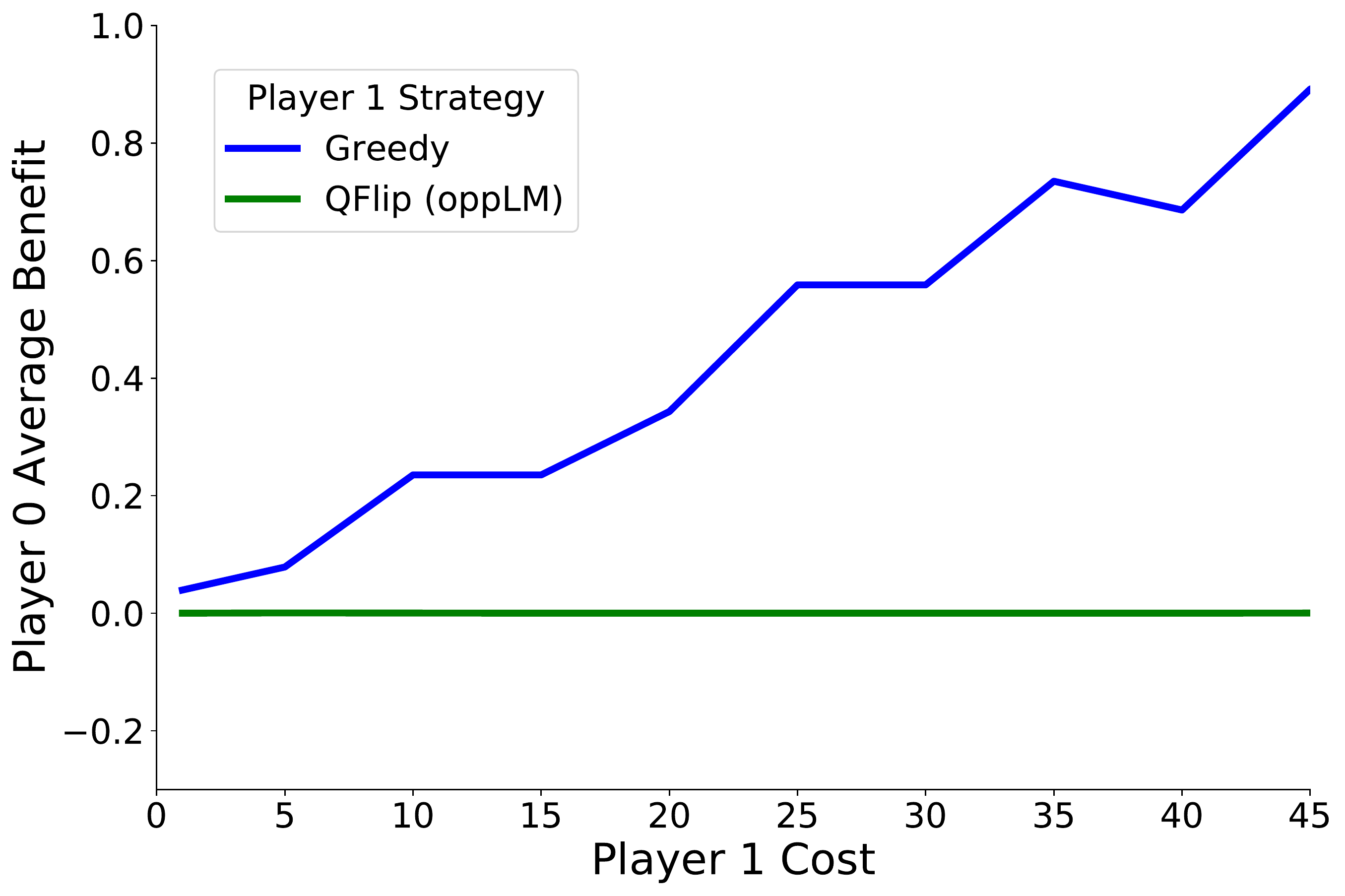}}
    \caption{Player 1 and Player 0's average benefit for \Q{} and Greedy across Player 1 costs. \Q{} with \oppLM{} playing against \Per{\delta} with fixed $\cost{0}=1$ and $\delta=50$ for 250,000 ticks, averaged over 100 runs.}
    \label{fig:greedy_vs_q}
\end{figure}


\subsection{\Q{} vs. Exponential}
The optimal \LM{} strategy against an \Exp{\lambda} opponent is proven in~\cite{FlipIt} to be \Per{\delta} with $\delta$ dependent on $\cost{0}$ and $\lambda$. The exponential distribution is memoryless, so optimal $\delta$ is independent of time since the \textit{opponent's} last move. Optimal \Q{} ignores $\tau$ and moves $\delta$ steps after its \textit{own} last move. \Q{} therefore prefers the \ownLM{} observation space from Table \ref{fig:state-types}, rather than \oppLM{} used against Periodic.


\begin{figure}[tbh]
  \subfloat[][\Q{} average benefit versus cost. \Q{} becomes near-optimal for all costs when duration is high.]
  {\includegraphics [width=.47\textwidth] {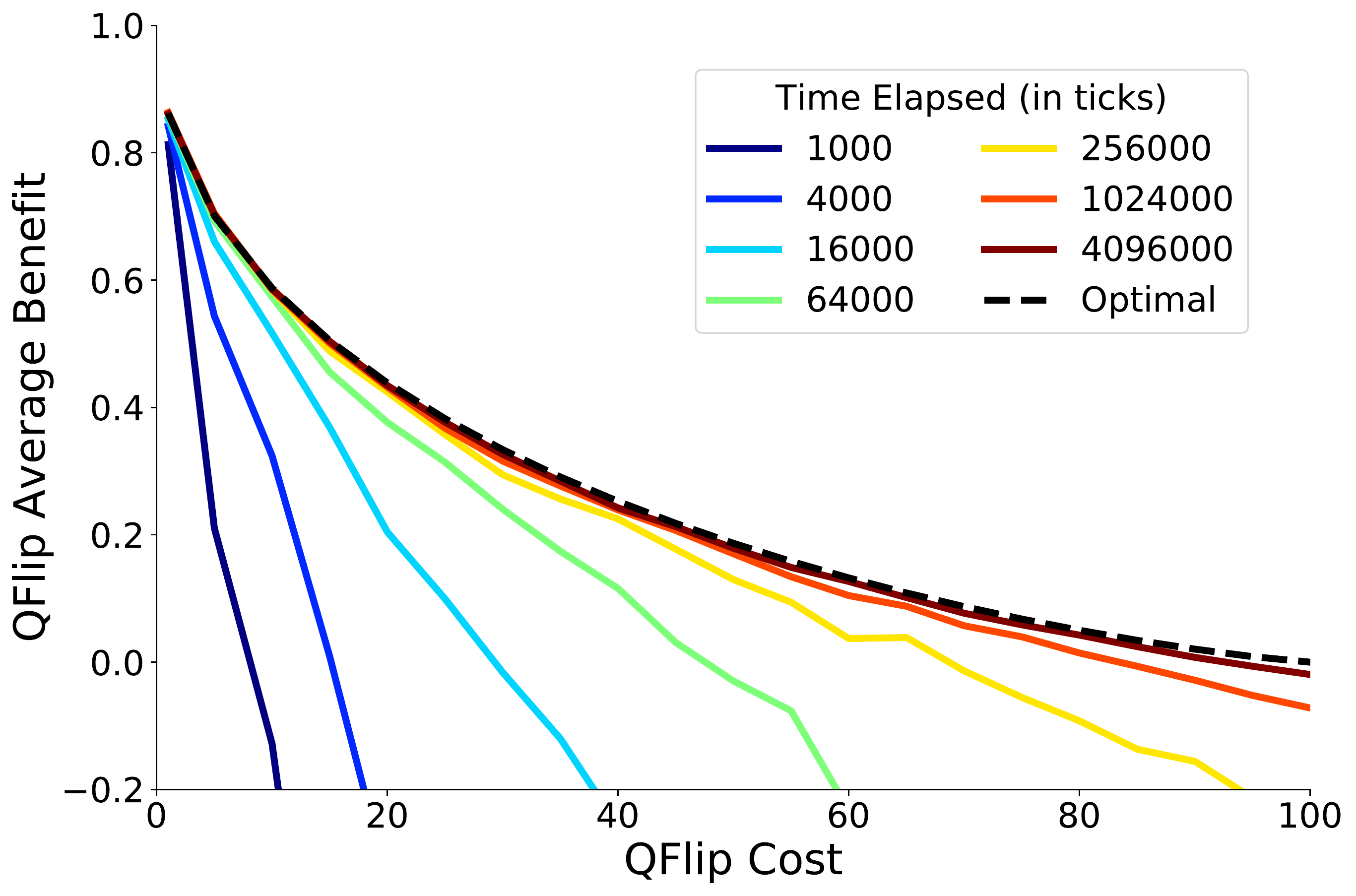}}
  \qquad
  \subfloat[][Ratio of \Q{} and Optimal average benefit versus time. \Q{} achieves optimal benefit quickly when costs are low. ]
  {\includegraphics [width=.47\textwidth] {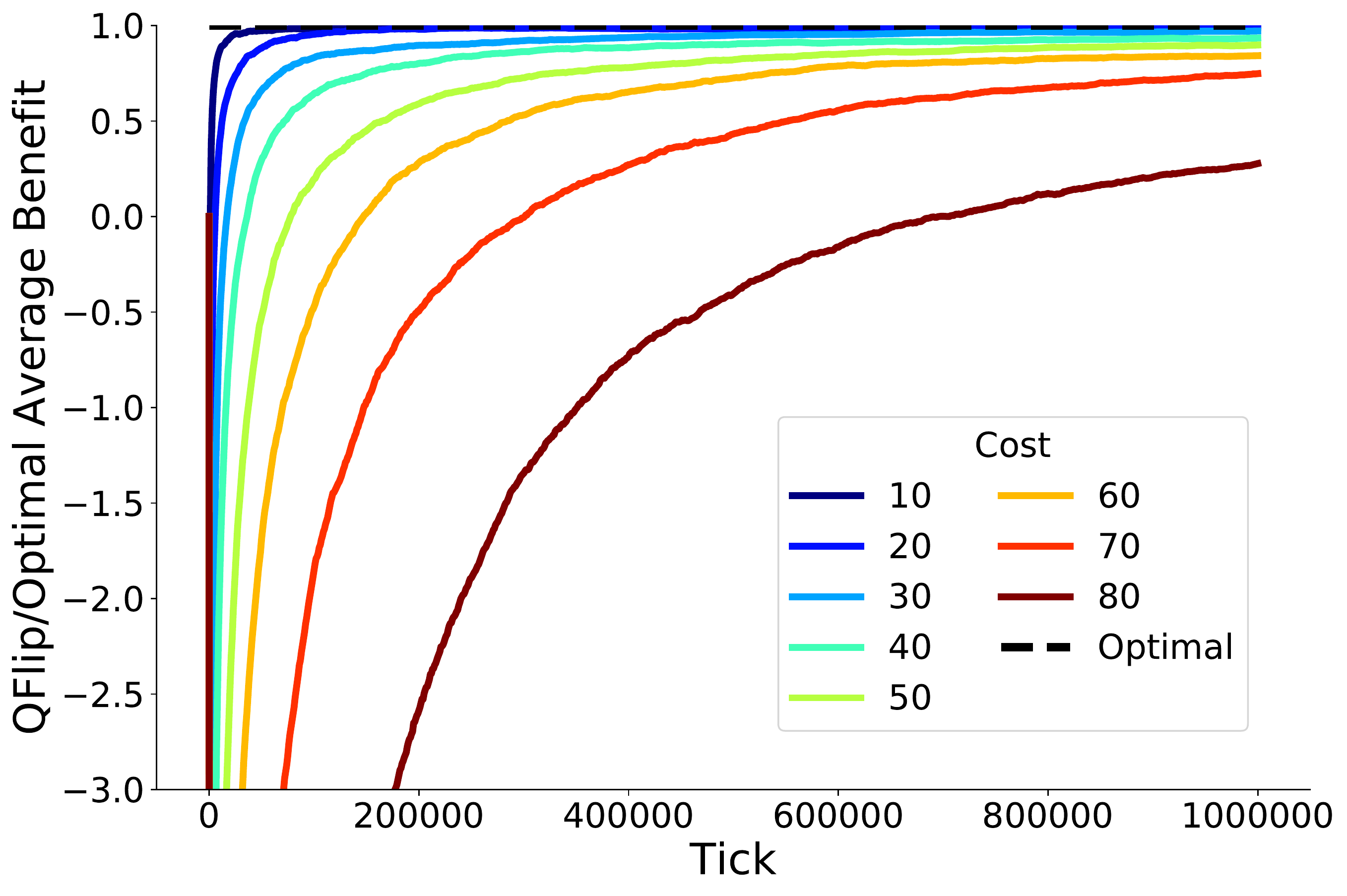}}
  \caption{\Q{} with \ownLM{} playing against \Exp{\lambda} with $\lambda=1/100$. }\label{fig:expected_dur}
\end{figure}

For \Q{} to learn any \Per{\delta} strategy, it must visit states $s<\delta$ many times. When playing against an Exponential opponent, the optimal $\delta$ grows quickly as $\cost{1}$ increases. For instance, against an \Exp{\lambda} opponent with $\lambda=1/100$ the optimal $\Per{\delta}$ strategy is $\delta=53$ for $\cost{1}=10$ and $\delta=389$ for $\cost{1}=90$~\cite{FlipIt}. 
As a result, the optimal Periodic strategy takes longer to learn as $\cost{1}$ grows.
 Figure \ref{fig:expected_dur} (a)  shows the average benefit versus cost after running the algorithm for up to 4.096 million ticks for rate of the Exponential distribution $\lambda=1/100$. For small costs, \Q{} learns to play optimally within a very short time ($16,000$ ticks). As the move cost increases, \Q{} naturally takes longer to converge. We verified this for other values of $\lambda$ as well.  Figure \ref{fig:expected_dur} (b) shows how the benefit varies by time for various move costs.  Given enough time, \Q{} converges to a near-optimal Periodic strategy for all costs (even as high as $\cost{1}=100$, which results in drop out for $\lambda=1/100$).

\section{Generalized \Q{} Strategy}
\label{sec:general}

Previous sections show that \Q{} converges to optimal using the \oppLM{} and \ownLM{} observation schemes for the \Per{\delta} and \Exp{\lambda} opponents respectively. In this section we show that \Q{} using a \composite{} observation scheme can play optimally against \Per{\delta} and \Exp{\lambda}, and perform well against other Renewal strategies without any knowledge of the opponent's strategy. The \composite{} strategy uses as states both Player 1's own last move time ($\lm{1}$), and the time since the opponent's last \textit{known} move ($\tau_t$), as described in Table~\ref{fig:state-types}. Composite \Q{} is the first general adaptive \flipit\ strategy that has no prior information on the opponent.


\subsection{Composite \Q{} Against \Per{\delta} and \Exp{\lambda} Opponents}
Figure \ref{fig:composite1} shows that \Q{}'s average benefit eventually converges to optimal against both \Exp{\lambda} and \Per{\delta} when using a \composite{} observation scheme. We note that it takes significantly longer to converge to optimal when using the \composite{} scheme. This is natural, as $\Q{}$ has an enlarged state space (quadratic compared to \oppLM{} and \ownLM{} observation schemes) and now visits each state less frequently. We leave approximation methods to expedite learning to future work.

\begin{figure}[tb]
    \subfloat[][\Q{} against \Per{\delta}with $\delta=50$.]
    {\includegraphics [width=.47\textwidth] {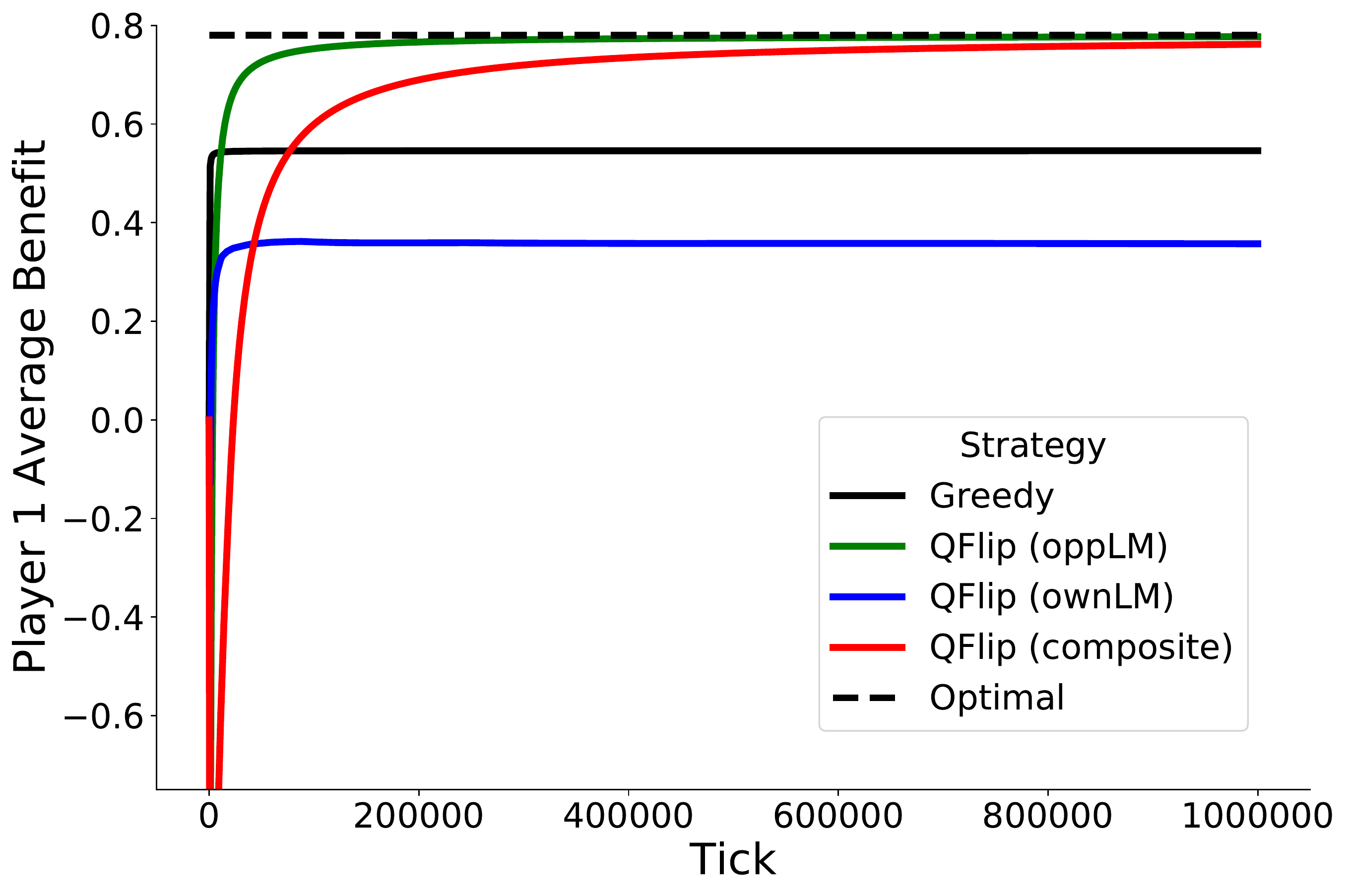}}
    \qquad
    \subfloat[][\Q{} against \Exp{\lambda} with $\lambda=1/100$.]
    {\includegraphics [width=.47\textwidth] {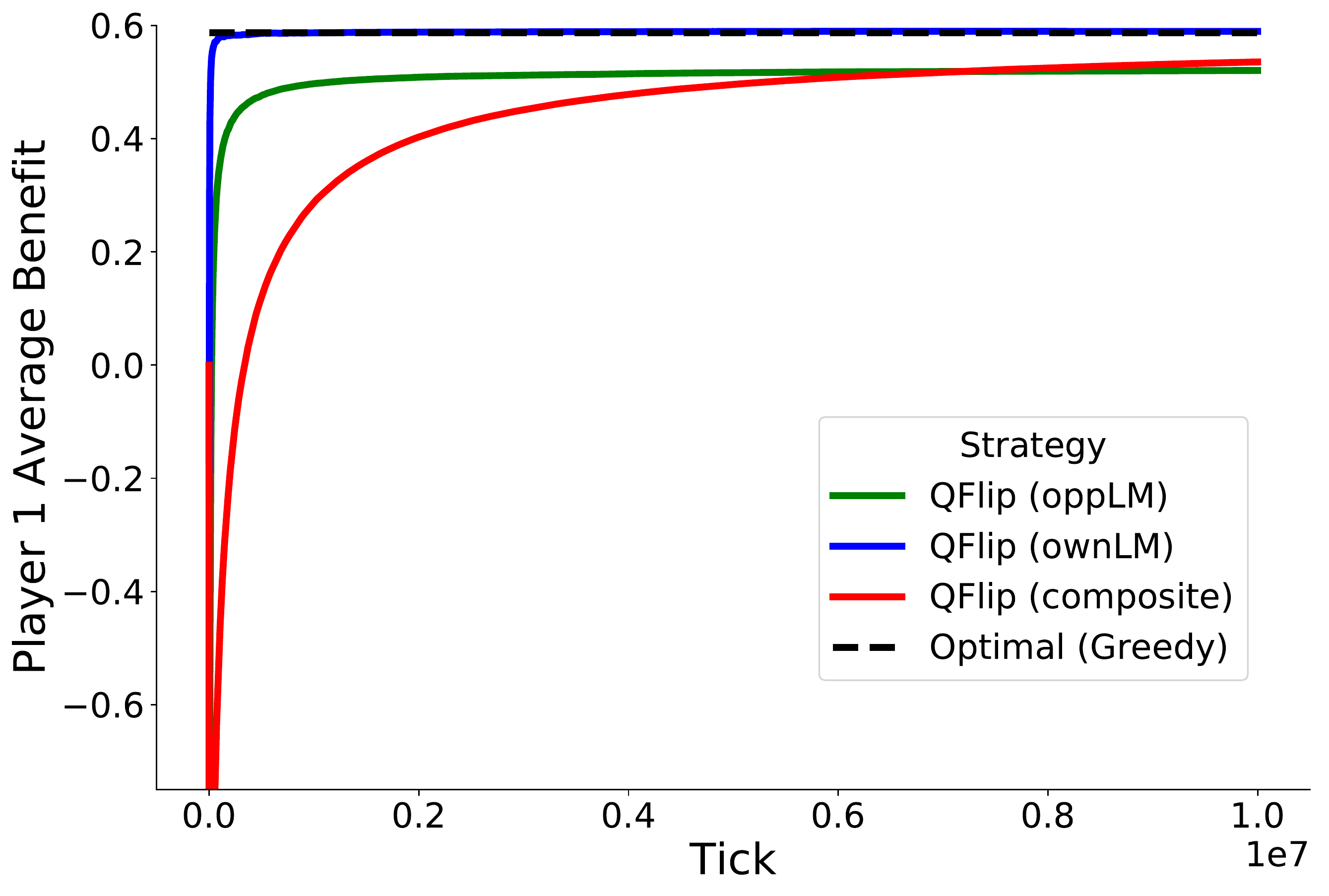}}
    \caption{\Q{} with $\cost{1}=10$ averaged over 10 runs. \Q{} converges to optimal against \Per{\delta} and  \Exp{\lambda} using \oppLM{} and \ownLM{} observation schemes respectively, and plays close to optimally against both with \composite{} observation scheme.}\label{fig:composite1}
\end{figure}
\begin{figure}[tb]

    \subfloat[][\Q{} vs. \Uni{\delta}{u} with $\delta=100$, $u=50$]
    {\includegraphics [width=.47\textwidth] {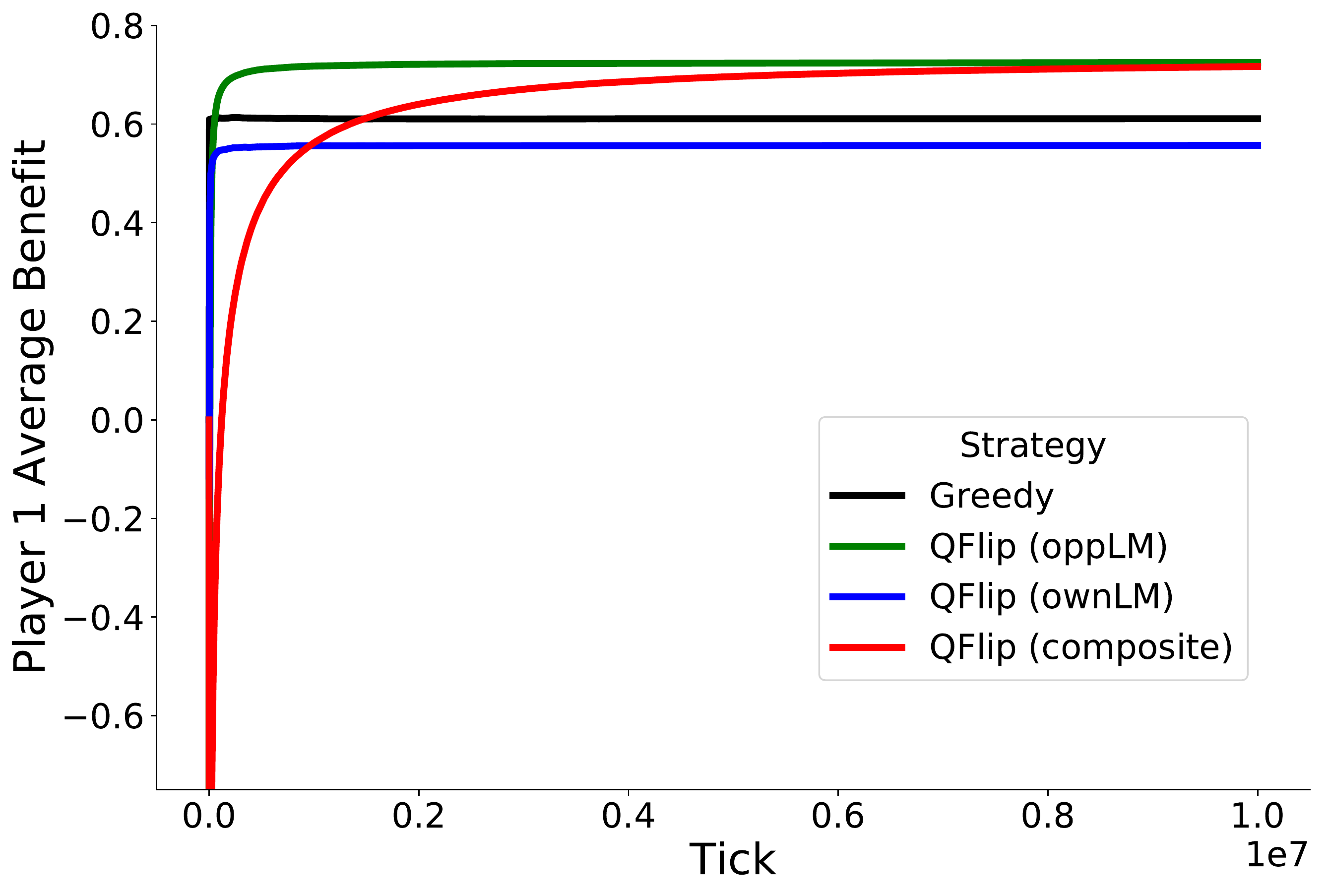}}
    \qquad
    \subfloat[][\Q{} vs. \Norm{\mu}{\sigma} with $\mu=100$, $\sigma=10$]
    {\includegraphics [width=.47\textwidth] {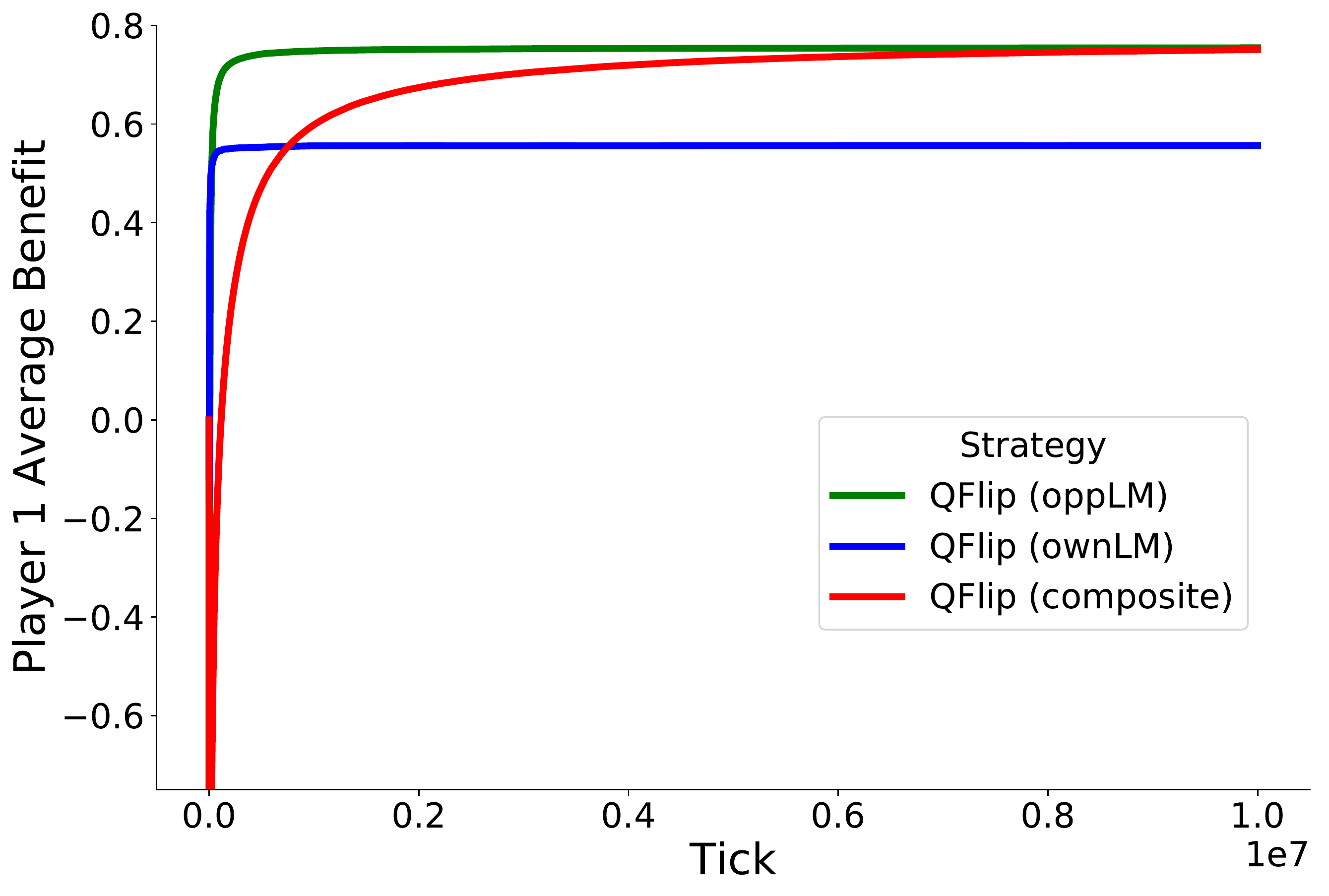}}

    \caption{\Q{}'s average benefit by time for Uniform (left) and Normal (right) distributions with $\cost{1}=10$. \Q{} with \oppLM{} and \composite{} observations outperforms Greedy against \Uni{\delta}{u} averaged over 10 runs. Against both opponents, \composite{} converges to \oppLM{} as time increases. }
    \label{fig:uniform}
\end{figure}
\subsection{Composite \Q{} Against Other Renewal Opponents}

The \composite{} strategy results in flexibility against multiple opponents. We evaluate \Q{} using \composite{} observations against Uniform and Normal Renewal opponents in Figure \ref{fig:uniform}. \Q{} attains 15\%  better average benefit than Greedy against \Uni{\delta}{u}. Figure \ref{fig:uniform} also shows that \Q{} attains a high average benefit of $0.76$ against a \Norm{\mu}{\sigma} opponent. We do not compare \Norm{\mu}{\sigma} to Greedy as the numerical packages we used were unable to find the maximum local benefit from Equation (\ref{eqn:localben}). \Q{} using \composite{} attains average benefit within 0.01 of \Q{} using \oppLM{} (best performing observation scheme) against both opponents after 10 million ticks.

\section{Conclusions}
\label{sec:conclusions}


We considered the problem of playing adaptively in the \flipit\ security game by designing \Q{}, a novel strategy based on temporal difference Q-Learning instantiated with three different observations schemes. We showed theoretically that \Q{} plays optimally against a Periodic Renewal opponent using the \oppLM{} observation. We also confirmed experimentally that \Q{} converges against Periodic and Exponential opponents, using the \ownLM{} observation scheme in the Exponential case. Finally, we showed general \Q{} with a \composite{} observation scheme performs well against Periodic, Exponential, Uniform, and Normal Renewal opponents. Generalized \Q{} is the first adaptive strategy which can play against any opponent with no prior knowledge.

We performed detailed experimental evaluation of our three observation \\ schemes for a range of distributions parameters and move costs. Interestingly, we showed that certain hyper-parameter configurations for the amount of exploration ($\epsilon$ and $d$), future reward discount ($\gamma$), and probability of moving in new states ($1-p$) are applicable against a range of Renewal strategies. Thus, \Q{} has the advantage of requiring minimal configuration. Additionally, we released an OpenAI Gym environment for \flipit{} to aid future researchers.


In future work, we plan to consider extensions of the \flipit\ game, such as multiple resources and different types of moves. We are interested in analyzing other non-adaptive strategies besides the class of Renewal strategies. Finally, approximation methods from reinforcement learning have the potential to make our composite strategy faster to converge and we plan to explore them in depth.

\section*{Acknowledgements}

We would like to thank Ronald Rivest, Marten van Dijk, Ari Juels, and Sang Chin for discussions about reinforcement learning in \flipit. We thank Matthew Jagielski, Tina Eliassi-Rad, and Lucianna Kiffer for discussing the theoretical analysis. This project was funded by NSF under grant CNS-1717634. This research was also sponsored by the U.S. Army Combat Capabilities Development Command Army Research Laboratory and was accomplished under Cooperative Agreement Number W911NF-13-2-0045 (ARL Cyber Security CRA). The views and conclusions contained in this document are those of the authors and should not be interpreted as representing the official policies, either expressed or implied, of the Combat Capabilities Development Command Army Research Laboratory or the U.S. Government. The U.S. Government is authorized to reproduce and distribute reprints for Government purposes notwithstanding any copyright notation here on.


%
%
\bibliographystyle{splncs04}
\bibliography{refs}
\ignore{

}

\end{document}